\documentclass[runningheads]{llncs}
\usepackage{graphicx}
\usepackage{epsfig}
\usepackage{array}
\usepackage{graphics}
\usepackage{amsmath}
\usepackage{algorithm}
\usepackage{algorithmic}

\usepackage{amssymb}

\usepackage{tikz}
\usepackage{epstopdf}
\usepackage{fullpage}
\usepackage{tikz}
\DeclareGraphicsExtensions{.pdf,.jpg,.png}

\date{}

\title{ON STRICTLY CHORDALITY-$k$ GRAPHS}
\author{ S.Dhanalakshmi and N.Sadagopan } 
\institute{Indian Institute of Information Technology, Design and Manufacturing, Kancheepuram, Chennai, India. \\
\email{$\{mat12d001,sadagopan\}@iiitdm.ac.in$}}

\begin{document}
\maketitle
\pagenumbering{arabic}
\pagestyle{plain}

\begin{abstract}
Strictly Chordality-$k$ graphs ($SC_k$ graphs) are graphs which are either cycle free or every induced cycle is exactly $k$, for some fixed $k, k \geq3$. Note that $k = 3$ and $k = 4$ are precisely the Chordal graphs and Chordal Bipartite graphs, respectively. In this paper, we initiate a structural and an algorithmic study of $SC_k, k \geq 5$ graphs. 
\end{abstract}

\section{Introduction}
The study of graphs with forbidden graph structures has attracted researchers from the field of mathematics and theory of computing. The popular ones are chordal and chordal bipartite graphs. Interestingly, these graphs find applications in computer architecture to factorize sparse matrix \cite{app1}, solving indefinite linear equations \cite{app2} and the study of linear programming \cite{app3}. A graph is chordal if every cycle of length at least 4 has a chord. Chordal graphs were introduced by Hajnal and Suranyi in 1958 \cite{Hajnal}. Dirac \cite{dirac} presented a structural characterization of chordal graphs with respect to minimal vertex separators and showed that chordal graphs are precisely the graph class in which every minimal vertex separator is a clique. A vertex is a \emph{simplicial vertex} if its neighborhood induces a clique. Interestingly, Dirac observed that every chordal graph has a simplicial vertex. Further, Fulkerson and Gross \cite{fulkerson} showed that all chordal graphs have a simplicial ordering (Perfect Elimination Ordering). On the time complexity front, chordal graphs can be recognized in polynomial time \cite{fulkerson,tarjan}.

Like chordal graphs, a related graph class, namely chordal bipartite graph received a considerable attention in the literature. A bipartite graph is chordal bipartite if every cycle of length at least 6 has a chord in it. Similar to chordal graphs, Golumbic and Goss \cite{GolumbicGoss} showed that a graph is chordal bipartite if and only if every minimal edge separator is a complete bipartite graph. Further, they can be recognized in polynomial time due to the existence of perfect edge elimination ordering \cite{GolumbicGoss}.

Both chordal and chordal bipartite graphs have received a good attention in the last four decades due to their nice structural and algorithmic characterizations. We also highlight that many classical combinatorial problems such as Vertex cover \cite{tarjan,gavril}, Clique cover \cite{corneil,Hoang}, Independent set \cite{gavril}, Treewidth \cite{bodlaender,kloks} are polynomial-time solvable when the input is restricted to chordal and chordal bipartite graphs, which are NP-Complete on general graphs. In some sense, these two graphs help to identify the gap between polynomial-time solvable input instances and the input instances that cause NP-Hardness. Other notable combinatorial problems such as Dominating-set \cite{booth,mueller}, Hamiltonian path \cite{colbourn,muellerh} remain NP-Complete on chordal and chordal bipartite graphs. It is important to highlight that chordal and chordal bipartite graphs are well studied graphs in the literature as it is clearly evident from some of the recent results on Join colorings \cite{join}, Contractibility problems \cite{contractibility}, Strong Chromatic index \cite{strong}, Enumeration of minimal dominating sets \cite{enumeration}, Reconfiguration graphs for vertex colourings \cite{reconfiguration} restricted to  chordal and chordal bipartite graphs.


A relook on the definition reveals that chordal graphs (chordal bipartite graphs) are graphs which are either cycle free or every induced cycle is $C_3$ (induced cycle is $C_4$ for chordal bipartite graphs). It is natural to ask, what is the graph class which are either cycle free or every induced cycle is $C_5$ and we call them as Strictly Chordality-5 graphs ($SC_5$ graphs). Interestingly, these graphs have the additional property that the girth (the length of the shortest cycle) equals the chordality (the length of the longest induced cycle). We shall explore this question in a larger dimension and initiate the study of \emph{Strictly Chordality-$k$ graphs} ($SC_k$ graphs), girth = chordality = $k$, for some $k \geq 3$. Thus, in this paper, we shall investigate a structural and an algorithmic study of $SC_k, k \geq 5$ graphs and we believe that this investigation has not been done in the literature.\\
\textbf{Our Contributions:} In the context of strictly chordality-$k$ graphs, $k \geq 5$, we show the following results: 
\begin{itemize}
\item[1.] Every minimal vertex separator in $SC_{2k+3}$ graphs, $k\geq 1$, is of cardinality at most two.
\item[2.] Every minimal vertex separator in $SC_{2k+4}$ graphs, $k\geq 1$, is of cardinality at most $s$, where $s$ is the size of the maximum \emph{cage}. 
\item[3.] We show that in every $SC_k$ graphs, there exists a special vertex or special $C_k$. Further, we show a special ordering among the vertices and cycles of $SC_k$.
\item[4.] Recognizing $SC_k$ graph can be done in polynomial-time.
\item[5.] We show that every $SC_k$ graphs, $k \geq 5$, is hamiltonian if and only if it is $2-$connected, $3$-$C_k$ pyramid free and $3$-cage free.
\item[6.] Every $SC_k$ graph, $k \geq 5$ is 2-colorable if $k$ is even and 3-colorable if $k$ is odd.
\item[7.] We establish that tree-width of $SC_k$ graphs is at most two.
\item[8.] We show that minimum fill-in problem is polynomial-time solvable.
\end{itemize}
\textbf{This paper is organized as follows:} We present graph preliminaries in Section 2. Structural observations on $SC_k, k \geq 5$ graphs based on minimal vertex separators are addressed in Section 3. We characterize $SC_k$ graphs by establishing an ordering in Section 4.  The algorithmic results like testing a graph, coloring, hamiltonicity, treewidth and minimum fill-in for $SC_k, k \geq 5$ graphs are presented in Section 5.
\section{Graph Preliminaries}

Notations used in this paper are as per \cite{golumbicbook,dbwest}. The graphs considered in this paper are simple, undirected, connected and unweighted. Let $G$ be a simple connected graph with the non-empty vertex set $V(G)$ and the edge set $E(G)$= \{\{$u,v$\} $\vert$ $u,v \in V(G)$ and $u$ is adjacent to $v$ in $G$ and $u \neq v$\}. The $neighborhood$ of a vertex $v$ of $G$, $N_G$($v$), is the set of vertices adjacent to $v$ in $G$. The degree of the vertex $v$ is $d_G(v) = \vert N_G(v) \vert$. $\delta(G)$ and $\Delta (G)$ denotes the minimum and maximum degree of a graph $G$, respectively. A graph $G$ is said to be $k$-$regular$ if $k = \delta (G) = \Delta (G)$. The graph $M$ is called a $subgraph$ of $G$ if $V(M)$ $\subseteq$ $V(G)$ and $E(M)\subseteq E(G)$. The subgraph $M$ of a graph $G$ is said to be $induced$ $subgraph$, if for every pair of vertices $u$ and $v$ of $M$, \{$u,v$\} $\in$ $E(M)$ if and only if \{$u,v$\} $\in$ $E(G)$ and it is denoted by $[M]$. $P_{uv} = (u=u_1, u_2, \ldots, u_k=v)$ is a \emph{path} defined on $V(P_{uv})=\{u=u_1, u_2, \ldots, u_k=v\}$ such that $E(P_{uv}) = \{\{u_i, u_{i+1}\}\vert \{u_i, u_{i+1}\} \in E(G), 1 \leq i \leq k-1\}$. For simplicity, we use $\vert P_{uv} \vert$ to refer to  $\vert V(P_{uv}) \vert$. The set $V(P_{uv})\backslash \{u,v\}$ denotes the \emph{internal vertices} of the path $P_{uv}$. $P_n$ denotes the path on $n$ vertices. A \emph{cycle} $C$ on $n$-vertices is denoted as $C_n$, where $V(C) = \{x_1, x_2, \ldots, x_n\}$ and $E(C) = \{\{x_1, x_2\}, \{x_2,x_3\}, \ldots, \{x_{n-1},x_n\}, \{x_n,x_1\}\}$. An $induced$ $cycle$ is a cycle that is an induced subgraph of $G$. A graph $G$ is said to be \emph{cycle-free} if there is no induced cycle in $G$. A graph $G$ is said to be $connected$ if every pair of vertices in $G$ has a path and if a graph is disconnected, it can be divided into disjoint connected $components$ $G_1, G_2, \ldots, G_k$, $k \geq 2$, where $V(G_i)$ denotes the set of vertices in the component $G_i$. Let $S$ be a non-empty subset of $V(G)$ and let $G\backslash S$ denotes the induced subgraph on $V(G)\backslash S$. The set $S$ is said to be an \emph{independent set} if every pair of vertices of $S$ is non-adjacent. For $\{u,v\} \notin E(G)$, a subset $R \subset V(G)$ is a $(u,v)$-vertex separator if $u$ and $v$ lies in different connected components of $G\backslash R$. $R$ is a minimal $(u,v)$-vertex separator if there does not exist a $(u,v)$-vertex separator $R'\subset R$. A vertex $v$ of a connected graph $G$ is said to be a \emph{cut vertex}, if $G\backslash \{v\}$ is a disconnected graph. An edge $e = \{u,v\}$ of a connected graph $G$ is said to be a \emph{cut-edge}, if the deletion of an edge $e$ from $G$ disconnects the graph $G$.

\section{Structural Observations on Strictly Chordality-$k$ Graphs}
Recall that, a graph $G$ is said to be a strictly chordality-$k$ graph, $SC_k$, if every induced cycle is of length exactly $k$ or $G$ is cycle-free. In this section, we present some structural observations on $SC_k, k \geq 5$, graphs with respect to minimal vertex separators.

\begin{lemma}
\label{intersection}
Let $G$ be a connected $SC_k, k\geq 5$, graph. For any two induced cycles $S_i$ and $S_j$ in $G$, one of the following is true.
\begin{itemize}
\item[(i)] $\vert V(S_i) \cap V(S_j) \vert \leq 1$
\item[(ii)] $\vert E(S_i) \cap E(S_j) \vert \leq 1$
\item[(iii)] $\mid E(S_i) \cap E(S_j) \mid = \frac{k}{2}$ if $k$ is even
\end{itemize}
\end{lemma}
\begin{proof}
On the contrary, assume that there exist induced cycles $S_i = (x_1, x_2, \ldots, x_k)$ and $S_j = (y_1,y_2, \ldots, y_k)$ such that $\vert V(S_i) \cap V(S_j) \vert \geq 2$ and $\vert E(S_i) \cap E(S_j) \vert \geq 2$ and, $k$ is even and $\mid E(S_i) \cap E(S_j) \mid \neq \frac{k}{2}$. The only possible cycles satisfying these condition's are; If $k$ is odd, then for every $3 \leq l \leq k-1$ and if $k$ is even, then for every $l \neq \frac{k}{2}+1$ and $l \in \{3,\ldots, k-1\}$,  $\vert V(S_i) \cap V(S_j) \vert = l$ and $\vert E(S_i) \cap E(S_j) \vert = l-1$.  i.e., there exist at least two cycles $S_i$ and $S_j$ in $G$ such that both contains a $P_l = (x_1, x_k, x_{k-1}, \ldots, x_{k-l+2}) = (y_1, y_k, y_{k-1}, \ldots, y_{k-l+2})$ in common (\emph{see Figure  \ref{sisj}}). Let $S$ be the set of internal vertices of $P_l$. The graph 
\begin{figure}[h]
\centering
\includegraphics[scale=0.35]{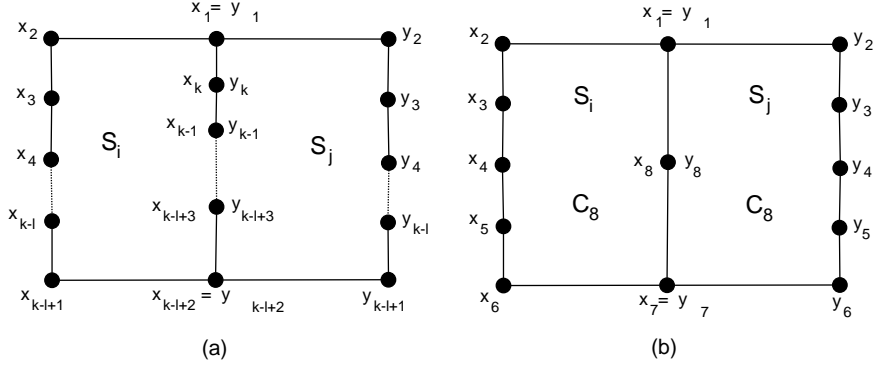}
\caption{(a) An illustration when $\vert E(S_i) \cap E(S_j) \vert = l-1$, where $l \neq \frac{k}{2}+1$ and $l \in \{3,\ldots, k-1\}$ (b) An example when $k=8$ and $l=3$}
\label{sisj}
\end{figure}
\noindent $G \backslash S$ induces $C_{2(k-l)+2}$. Note that, the cycle is induced because any chord from $x_i$ to $y_j$, $i,j \in \{2,3,\ldots,k-l+1\}$ induces either $C_{2k-2l-i-j+5}$ or $C_{i+j-1}$, for any $l \geq 3$ and $k \geq 5$. Since $4 \leq i+j \leq 2k-2l+2$, neither $C_{2k-2l-i-j+5}$ nor $C_{i+j-1}$ is $C_k$,  for any $l \geq 3$ and $k \geq 5$, which contradicts the definition of $SC_k$ graphs and hence, the lemma follows.$\hfill \qed$
\end{proof}
                                                                                                                                                                               
                                                                                                                                                                               \noindent Note that the induced cycles $S_i$ and $S_j$ in an $SC_k$ graph $G$ is said to have \emph{vertex intersection} if $\vert V(S_i) \cap V(S_j) \vert = 1$ and \emph{edge intersection} if $\vert E(S_i) \cap E(S_j) \vert = 1$.
                                                                                                                                                                                                                                                                                                                                                               \begin{corollary}                                                                                                                                                                                                                                                                                                                                                                                                                                                                                                                                       \label{intersectionseparate} Let $G$ be a connected $SC_k$ graph, $k \geq 5$. For any two induced cycles $S_i$ and $S_j$ in $G$,  either $\vert V(S_i) \cap V(S_j) \vert \leq 1$ or $\vert E(S_i) \cap E(S_j) \vert \leq 1$, if $k$ is odd and either $\vert V(S_i) \cap V(S_j) \vert \leq 1$ or $\vert E(S_i) \cap E(S_j) \vert$ is $0$ or $1$ or $\frac{k}{2}$, if $k$ is even.                                                                                                                                                                                   \end{corollary}                                                                                                                                                                                                                                                                                                                                                            \begin{proof}                                                                                                                                                                              Trivially follows from \emph{Lemma \ref{intersection}}. $\hfill \qed$                                                                                                                                                                                  \end{proof}

This corollary acts as a powerful tool to determine the maximum size of the minimal vertex separator in an $SC_k$ graph as well as the structure of minimal vertex separators in $SC_k$ graphs which we shall present next.                                                                                                                                                                                 
                                                                                                                                                                                 
\begin{theorem}
\label{mvssc2k+1}
 Let $G$ be a connected $SC_{k}$ graph, $k = 2m+3, m \geq 1$. The cardinality of every minimal vertex separator of $G$ is at most 2.
\end{theorem}
\begin{proof} On the contrary, assume that there exist a minimal vertex separator $S$ such that $\vert S \vert$ = $n$, $n \geq$ 3. The graph $G\backslash S$ is a disconnected graph with distinct connected components $G_1, G_2, \ldots, G_l$, $l \geq 2$. 
Consider the graph $H$ induced on the set $V(H) = V(G_1) \cup V(G_2) \cup S$. Throughout this proof, when we refer to $P_{xy}^i$, we mean the shortest path $P_{xy}$ where every internal vertex belongs to $G_i$, $1 \leq i \leq 2$. Let $t, u$ and $v$ be any three vertices in $S$ and let $S' = \{t, u, v\}$. Since $S$ is a minimal vertex separator, every vertex in $S$ is adjacent to at least one vertex in each component. Thus, for every pair $x, y \in S'$ there exists $P_{xy}^1$ and $P_{xy}^2$ ($\because G_1$ and $G_2$ are connected components of $H \backslash S$). Let $P_{tu}^{1}$ = ($t, a=a_1, \ldots, b=a_p, u$), $P_{tu}^2$ = ($t, w=w_1, \ldots, x=w_q, u$), $P_{tu}^{1}$ = ($u, c=c_1, \ldots, d=c_{r}, v$), $P_{uv}^2$ = ($u, y=y_1, \ldots, z=y_{s}, v$), $P_{bc}^{1}$ = ($b, b_1, \ldots, c$) and $P_{xy}^{2}$ = ($x, x_1, \ldots, y$). Note that if $b\neq c$, then $(v,P_{bc}^{1})$ forms an induced $C_k$ and if $x \neq y$, then $(v,P_{xy}^{2})$ forms an induced $C_k$. We complete this proof using case analysis (\emph{see Table \ref{table:table1}}) by considering the cases where $S'$ is independent and not independent.

In each case, we arrive at a contradiction by exhibiting an induced cycle other than $C_k$. Further, we exhibit two induced cycles $S_i$ and $S_j$ with $P_n, n \geq 3$ in common, which contradicts \emph{Corollary \ref{intersectionseparate}}. It follows that our assumption that there exist a minimal vertex separator of size 3 or more is wrong. Thus, the theorem is true for $H$ and hence the super graph $G$ as every induced cycle $H$ is also an induced cycle in $G$. $\hfill \qed$
\end{proof}

\begin{figure}[h]
\vspace{-0.5cm}
\begin{center}
\includegraphics[scale=0.25]{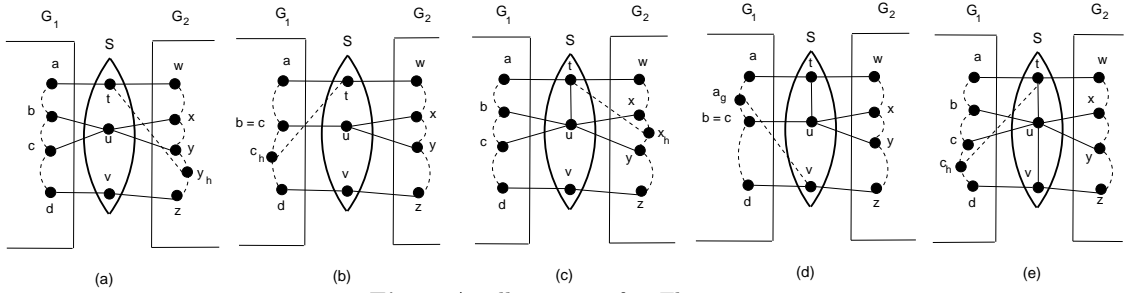}
\vspace{-0.5cm}
\caption{An illustration for \emph{Theorem \ref{mvssc2k+1}}}
\label{mvsk1}
\end{center}

\end{figure}

\vspace{-1.4cm}
\begin{table}[h]
\centering
\label{pre}
\caption{Possible chords from vertices $t$ and $v$}
\begin{tabular}{|l|l|l|l|} \hline
\textbf{Type A:} & $t$ is adjacent to a vertex in $V(P_{bc}^{1})\backslash \{b\}$ &
\textbf{Type E:} & $v$ is adjacent to a vertex in $V(P_{bc}^{1})\backslash \{c\}$ \\
\textbf{Type B:} & $t$ is adjacent to a vertex in $V(P_{cd}^{1})\backslash \{c\}$ &
\textbf{Type F:} & $v$ is adjacent to a vertex in $V(P_{ab}^{1})\backslash \{b\}$ \\
\textbf{Type C:} & $t$ is adjacent to a vertex in $V(P_{xy}^{2})\backslash \{x\}$ &
\textbf{Type G:} & $v$ is adjacent to a vertex in $V(P_{xy}^{2})\backslash \{y\}$ \\
\textbf{Type D:} & $t$ is adjacent to a vertex in $V(P_{yz}^{2})\backslash \{y\}$ &
\textbf{Type H:} & $v$ is adjacent to a vertex in $V(P_{wx}^{2})\backslash \{x\}$ \\ \hline
\end{tabular}
\vspace{-0.8cm}
\end{table}

\begin{table}[H]
\caption{Case analysis for the proof of Theorem \ref{mvssc2k+1}}
\label{table:table1}
\begin{tabular}{|l|l|}
\multicolumn{2}{l}{\textbf{Case 1:}  $[S']$ is independent. $(P_{tu}^{1},P_{tu}^2)$ and $(P_{uv}^{1}, P_{uv}^2)$ form an induced $C_k$ in $H$.}\\ \hline

\multicolumn{1}{|c|}{\textbf{Case Analysis}} & \multicolumn{1}{|c|}{\textbf{Induced cycles with justification}}\\  \hline 
\textbf{Case 1.1:} $b \neq c$ and $x \neq y$; &  $(t, P_{ab}^{1}, P_{bc}^{1}, P_{cd}^{1}, v, P_{yz}^{2}, P_{xy}^{2}, P_{wx}^{2})$ forms an induced $C_{n >k}$, a contradiction. \\ 
& 	The cycle is induced by the following sub cases:\\ \cline{2-2} 
&\textbf{Case 1.1a:} Chord of \emph{Type C}. $S_i=(P_{tu}^{1}, P_{xx_h}^{2})$ and $S_j = (u,P_{xy}^{2})$, where, $x_h$ is the least indexed \\
& vertex in $P_{xy}$ such that $\{t,x_h\} \in E(G)$; Note that $V(S_i) = V(P_{tu}^{1}) \cup V(P_{xx_h}^{2})$ and \\
& $V(S_j) = \{u\} \cup V(P_{xy}^{2})$. $[V(S_i)\cap V(S_j)] = P_n, n \geq 3$, a contradiction. \\ 
&	Similar arguments hold good for chords of \emph{Type A, G} and \emph{E}. \\ \cline{2-2}
&\textbf{Case 1.1b:} Chord of \emph{Type D}. $S_i=(P_{tu}^{1}, P_{yy_h}^{2})$ and $S_j = (P_{uv}^{1},P_{uv}^{2})$, where, $y_h$ is the least indexed \\
&  vertex in $P_{yz}$ such that $\{t,y_h\} \in E(G)$; $[V(S_i)\cap V(S_j)] = P_{n \geq 3}$, a contradiction (\emph{see Figure 2(a)}). \\ 
&	Similar arguments can be given if chords are of \emph{Type B, F} and \emph{H}.\\
\hline
\textbf{Case 1.2:} $b = c$ and $x \neq y$; & $(t, P_{ab}^{1}, P_{cd}^{1}, v, P_{yz}^{2}, P_{xy}^{2}, P_{wx}^{2})$ forms an induced $C_{n >k}$, a contradiction. \\ 
& The cycle is induced by the following sub cases:\\ \cline{2-2}

&	\textbf{Case 1.2a:} Chord of \emph{Type B}. $S_i=(P_{tu}^{2}, P_{cc_h}^{1})$ and $S_j = (P_{uv}^{1},P_{uv}^{2})$, where, $c_h$ is the least indexed \\
&	vertex in $P_{cd}$ such that $\{t,c_h\} \in E(G)$; $[V(S_i)\cap V(S_j)] = P_{n \geq 3}$, a contradiction (\emph{see Figure 2(b)}).  \\ 
&	The argument is symmetric for chords of \emph{Type F}.\\ \cline{2-2}
& \textbf{Case 1.2b:} Chord of \emph{Type A} or \emph{G}.  The argument is similar to the \emph{Case 1.1a}. \\ \cline{2-2}
& \textbf{Case 1.2c:} Chord of \emph{Type D} or \emph{H}. The argument is similar to the \emph{Case 1.1b}. \\  \hline

\textbf{Case 1.3:} $b \neq c$ and $x = y$; & $(t, P_{ab}^{1}, P_{bc}^{1}, P_{cd}^{1}, v, P_{yz}^{2}, P_{wx}^{2})$ forms an induced $C_{n >k}$, a contradiction. \\
& The argument for the cycle is induced is symmetric to the \emph{Case 1.2}\\ \hline

\textbf{Case 1.4:} $b = c$ and $x = y$ & The induced cycles $S_i= (P_{tu}^{1},P_{tu}^{2})$ and $S_j=(P_{uv}^{1},P_{uv}^{2})$ have $P_n, n \geq 3$ \\
& in common, a contradiction.\\ \hline
\end{tabular}
\vspace{-0.8cm}
\end{table}

\begin{table}[H]
\begin{tabular}{|l|l|}
\multicolumn{2}{l}{\textbf{Case 2:} $[S']$ is not independent and $\{t,u\} \in E(G)$, $\{u,v\},\{t,v\} \notin E(G)$.} \\ \hline
\textbf{Case 2.1:} $b\neq c$ and $x \neq y$ & $(t, P_{ab}^{1}, P_{bc}^{1}, P_{cd}^{1}, v, P_{yz}^{2}, P_{xy}^{2}, P_{wx}^{2})$ forms an induced $C_{n >k}$, a contradiction. \\
&	The cycle is induced by the following sub cases:\\ \cline{2-2}

&	\textbf{Case 2.1a:} Chord of \emph{Type C}. $S_i=(t,P_{xx_h}^{2},u)$ and $S_j = (u,P_{xy}^{2})$, where, $x_h$ is the least indexed \\
& vertex in $P_{xy}$ such that $\{t,x_h\} \in E(G)$; $[V(S_i)\cap V(S_j)] = P_{n \geq 3}$, a contradiction (\emph{see Figure 2(c)}). \\
&	Similar argument hold good for chord of \emph{Type A}. \\ \cline{2-2}
&	\textbf{Case 2.1b} Chords of \emph{Type E} or \emph{G}.	The argument is similar to the \emph{Case 1.1a}. \\ \cline{2-2}
&	\textbf{Case 2.1c:} Chord of \emph{Type B}. $S_i=(t,P_{cc_h},u)$ and $S_j=(P_{uv}^{1},P_{uv}^{2})$, where, $c_h$ is the least indexed \\
 & vertex in $P_{cd}$ such that $\{t,c_h\} \in E(G)$; $[V(S_i)\cap V(S_j)] = P_n, n \geq 3$, a contradiction. \\ 
 & The argument is symmetric for chord of \emph{Type D}.\\ \cline{2-2}
 & \textbf{Case 2.1d:} Chord of \emph{Type F}.  $S_i=(t,P_{tu}^{1},u)$ and $S_j=(P_{uv}^{2},P_{aa_g}^{1},t)$, where, $a_g$ is the largest \\
 & indexed vertex in $P_{ab}$ such that $\{v,a_g\} \in E(G)$; $[V(S_i)\cap V(S_j)] = P_{n \geq 3}$, a contradiction.\\ 
 & Similar argument for \emph{Case H}.\\ \hline
\textbf{Case 2.2:} $b= c$ and $x \neq y$ & $(t, P_{ab}^{1}, P_{cd}^{1}, u, P_{yz}^{2}, P_{xy}^{2}, P_{wx}^{2})$ forms an induced $C_{n >k}$, a contradiction. \\ 
&The cycle is induced by the following sub cases:\\ \cline{2-2}

& \textbf{Case 2.2a:} Chord of \emph{Type F}. $S_i=(P_{uv}^{2},P_{ba_g}^{1})$ and $S_j = (t,P_{tu}^{1},u)$, where, $a_g$ is the largest indexed \\ 
& vertex in $P_{ab}$ such that $\{v,a_g\} \in E(G)$; $[V(S_i)\cap V(S_j)] = P_{n \geq 3}$, a contradiction (\emph{see Figure 2(d)}). \\ \cline{2-2}
& \textbf{Case 2.2b:} Chord of \emph{Type B}.  $S_i=(t,P_{cc_h}^{1},u)$ and $S_j=(P_{uv}^{1},P_{uv}^{2})$, where, \\ 
 & $c_h$ is the least indexed vertex in $P_{cd}$ such that $\{t,c_h\} \in E(G)$. \\ \cline{2-2}
&\textbf{Case 2.2c} Chords of \emph{Type C} or \emph{D} or \emph{G} or \emph{H:}\\
&  The arguments are similar to the sub cases of \emph{Case 2.1}. \\ \hline
\textbf{Case 2.3:} $b \neq c$ and $x = y$ & $(t, P_{ab}^{1}, P_{bc}^{1}, P_{cd}^{1}, v, P_{yz}^{2}, P_{wx}^{2})$ forms an induced $C_{n >k}$, a contradiction. \\
  & The argument is similar to the \emph{Case 2.2}\\ \hline
\textbf{Case 2.4:} $b = c$ and $x = y$ & $(t, P_{ab}^{1}, P_{cd}^{1}, v, P_{yz}^{2}, P_{wx}^{2})$ forms an induced $C_{n >k}$, a contradiction. \\ 
& The cycle is induced by the arguments in \emph{Case 2.2a; Case 2.2b; Case 2.3}\\ \hline
%
\multicolumn{2}{|l|}{\textbf{Case 3:} $[S']$ is not independent and $\{t,v\} \in E(G)$, $\{u,v\},\{t,u\} \notin E(G)$. The argument is similar to \emph{Case 2}.} \\ \hline

\multicolumn{2}{|l|}{\textbf{Case 4:} $[S']$ is not independent and $\{u,v\} \in E(G)$, $\{t,v\},\{t,u\} \notin E(G)$. The argument is similar to \emph{Case 2}.} \\ \hline

\multicolumn{2}{l}{\textbf{ }} \\
\multicolumn{2}{l}{\textbf{Case 5:} $[S']$ is not independent and $\{t,u\},\{u,v\} \in E(G)$, $\{t,v\} \notin E(G)$ }\\ \hline
\textbf{Case 5.1:} $b \neq c$ and $ x \neq y$ & $(t, P_{ab}^{1}, P_{bc}^{1}, P_{cd}^{1}, v, P_{yz}^{2}, P_{xy}^{2}, P_{wx}^{2})$ forms an induced $C_{n >k}$, a contradiction. \\ \hline
& The cycle is induced by the following sub cases:\\ \cline{2-2}

& \textbf{Case 5.1a:} Chord of \emph{Type A}.  The argument is similar to the \emph{Case 2.1a.} \\ & Similar arguments can be given if chords are of \emph{Type C, E} and \emph{G}. \\ \cline{2-2}
& \textbf{Case 5.1b:} Chord of \emph{Type B}. $S_i=(t,P_{cc_h}^{1},u)$ and $S_j=(u,P_{uv}^{1},v)$, where, $c_h$ is the least indexed \\ 
 &vertex in $P_{cd}$ such that $\{t,c_h\} \in E(G)$; $[V(S_i)\cap V(S_j)] = P_{n \geq 3}$, a contradiction (\emph{see Figure 2(e)}). \\ 
 & Similar arguments hold good for chords of \emph{Type D, F} and \emph{H}. \\ \hline  
\textbf{Case 5.2:} $b=c$ and $x \neq y$ &  $(t, P_{ab}^{1}, P_{cd}^{1}, v, P_{yz}^{2}, P_{xy}^{2}, P_{wx}^{2})$ forms an induced $C_{n >k}$, a contradiction. \\
& The argument is similar to the \emph{Case 5.1}.\\
 \hline
\textbf{Case 5.3:} $b \neq c$ and $x = y$ &  $(t, P_{ab}^{1}, P_{bc}^{1}, P_{cd}^{1}, v, P_{yz}^{2}, P_{wx}^{2})$ forms an induced $C_{n >k}$, a contradiction. \\
& The argument is similar to the \emph{Case 5.1}\\
 \hline
\textbf{Case 5.4:} $b = c$ and $x = y$ &  $(t, P_{ab}^{1}, P_{cd}^{1}, v, P_{yz}^{2}, P_{wx}^{2})$  forms an induced $C_{n >k}$, a contradiction. \\ 
& The argument is similar to the \emph{Case 5.1b}.\\
 \hline
\multicolumn{2}{|l|}{\textbf{Case 6:} $[S']$ is not independent and $\{t,u\},\{t,v\} \in E(G)$, $\{u,v\} \notin E(G)$. The argument is similar to \emph{Case 5}. }\\ \hline
\multicolumn{2}{|l|}{\textbf{Case 7:} $[S']$ is not independent and $\{t,v\},\{u,v\} \in E(G)$, $\{u,t\} \notin E(G)$. The argument is similar to \emph{Case 5}. }\\ \hline
\end{tabular}
\end{table}

\begin{lemma}
\label{mvssck}
Let $G$ be a connected $SC_k$ graph, $k=2m+4, m \geq 1$. For any two induced cycles $S_i$ and $S_j$: if either $\vert V(S_i) \cap V(S_j) \vert \leq 1$ or $\vert E(S_i) \cap E(S_j) \vert \leq 1$, then the cardinality of every minimal vertex separator of $G$ is at most 2.
\end{lemma}
\begin{proof}
An argument similar to \emph{Theorem \ref{mvssc2k+1}} establishes this claim. $\hfill \qed$
\end{proof}

\begin{lemma}
\label{mvsindependent}
Let $G$ be a connected $SC_k$ graph, $k=2m+4, m \geq 1$. If $S$ is a minimal vertex separator of $G$ with $\mid S\mid \geq 3$, then $S$ is an independent set.
\end{lemma}
\begin{proof}
On the contrary, assume that there exists a minimal vertex separator $S$ such that $\mid S\mid \geq 3$ and $S$ is not an independent set. Let $G_1, G_2, \ldots, G_l$, $l\geq 2$ be the connected components of $G\backslash S$. Consider the graph $H$ induced on the set $V(H) = V(G_1) \cup V(G_2) \cup S$. Choose any three vertices, $S'=\{t,u,v\}$, from $S$ such that either $\{t,u\} \in E(G)$ and $\{u,v\},\{t,v\} \notin E(G)$ or $\{t,u\},\{u,v\} \in E(G)$ and $\{t,v\} \notin E(G)$. Since $S$ is a minimal vertex separator, every vertex in $S$ is adjacent to at least one vertex in each component. Thus, $P^{1}_{tu}$ and $P^{2}_{tu}$ exists and these paths create a cycle of length $k$, say $S_1=P^{1}_{tu}=(t,a_1,a_2,\ldots,a_{k-2},u)$ and $S_2=P^{2}_{tu}=(t,w_1,w_2,\ldots,w_{k-2},u)$. Let $b_1$ be a vertex in $G_1$ which is adjacent to $v$ in $P^{1}_{va_{k-2}}$ and $x_1$ be a vertex in $G_2$ which is adjacent to $v$ in $P^{2}_{vw_{k-2}}$.
\begin{description}

\item[\textbf{Case ($\{t,u\} \in E(G)$} and \textbf{$\{u,v\},\{t,v\} \notin E(G)$):}] 
 It is clear that, $S_3 = (P^{1}_{va_{k-2}},P^{2}_{vw_{k-2}},u)$ forms an induced $C_k$. Let $\mid V(P^{2}_{x_1w_{k-2}}) \mid = n$ and $\mid V(P^{1}_{b_1a_{k-2}}) \mid = k$-$n$-$2$. Thus, $P^{1}_{va_{k-2}}=(v,b_1,\ldots,b_{k-n-2}=a_{k-2})$ and $P^{2}_{vw_{k-2}}=(v,x_1,\ldots,x_n=w_{k-2})$. Hence, $(P_{ta_{k-2}}^{1}, P_{va_{k-2}}^{1}, P_{vw_{k-2}}^{2}, P_{tw_{k-2}}^{2})$ forms an induced cycle of length greater than $k$. The cycle is induced because the following cases are not possible by the definition of $SC_k$.

\begin{figure}[h]
\centering
\includegraphics[scale=0.3]{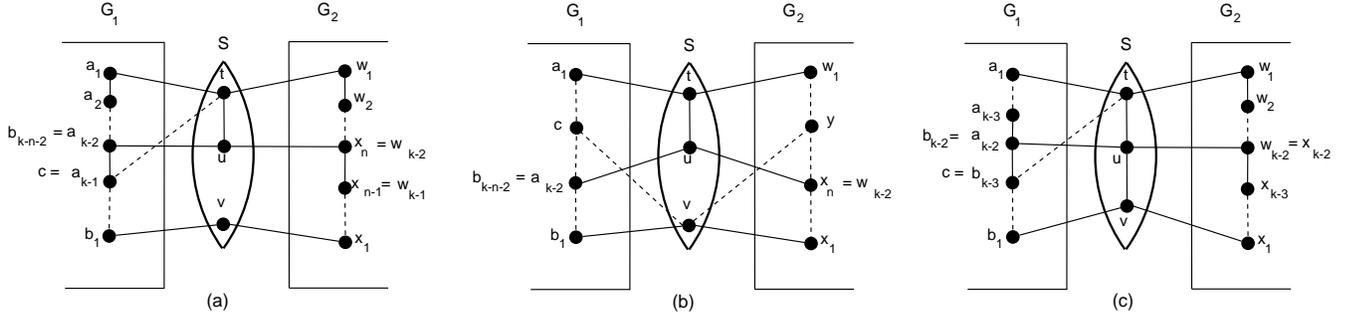}
\caption{An illustration of the graph when (a) $\{t,a_{k-1}\}\in E(G)$ and $\{u,v\} \notin E(G)$, (b) $\{v,c\}, \{v,y\}\in E(G)$ and $\{u,v\} \notin E(G)$, and (c) $\{t,c\} \in E(G)$ and $\{u,v\} \in E(G)$ }
\label{fig:is1}
\end{figure}

\begin{itemize}
\item[$\bullet$]  If $a_p$, $p \in \{1,\ldots,k-3\}$, has adjacency in $P^{1}_{b_1b_{k-n-3}}$, and $b_q$, $q \in\{1,\ldots, k-n-3\}$, has adjacency in $P^{1}_{a_1a_{k-3}}$. Choose the least $p$ such that $\{a_p,c\}\in E(G)$, $c\in \{b_1,\ldots, b_{k-n-3}\}$, and $(c, \ldots, a_{k-2}, \ldots, a_p)$ forms an induced $C_k$. Choose the least $q$ such that $\{b_q,d\}\in E(G)$, $d\in \{a_1,\ldots, a_{k-3}\}$ and $(d, \ldots, a_{k-2},$ $\ldots, b_q)$ forms an induced $C_k$. Then, either $(t,a_1,\ldots, a_p, c, \ldots, b_{k-n-2},u)$ or $(P^{2}_{uv}, v, b_1, \ldots, b_q, d, \ldots,$ $a_{k-2}, u)$ forms an induced cycle of length greater than $k$. 

\item[$\bullet$]  If $w_p$, $p \in \{1,\ldots,k-3\}$, has adjacency in $P^{2}_{x_1x_{n-1}}$ and $x_r$, $r \in\{1,\ldots, n-1\}$, has adjacency in $P^{2}_{w_1w_{k-3}}$. Choose the least $p$ such that $\{w_p,y\}\in E(G)$, $y\in \{x_1,\ldots, x_{n-1}\}$, and $(y, \ldots, w_{k-2}, \ldots, w_p)$ forms an induced $C_k$. Choose the least $r$ such that $\{x_r,z\}\in E(G)$, $z\in \{w_1,\ldots, w_{k-3}\}$ and $(z, \ldots, w_{k-2},$ $\ldots, x_r)$ forms an induced $C_k$. Then, either $(t,w_1,\ldots, w_p, y, \ldots, x_{n-1},u)$ or $(P^{1}_{uv}, v, x_1, \ldots, x_r, z, \ldots,$ $ w_{k-2}, u)$ forms an induced cycle of length greater than $k$. 

\item[$\bullet$] If $t$ is adjacent to some vertices in $P^{1}_{b_1b_{k-n-2}}$. Pick the largest indexed vertex in $P^{1}_{b_1b_{k-n-3}}$, say $c$, such that $t$ is adjacent to $c$. If $c=a_{k-1}=b_{k-n-3}$, then $(t,a_{k-1},a_{k-2},u)$ forms an induced $C_4$ (\emph{see Figure  \ref{fig:is1}(a)}). If $c \in \{b_1,b_2,\ldots,b_{k-n-4}\}$, then $(t,a_1,\ldots,a_{k-2}=b_{k-n-2}, \ldots, c)$ creates an induced cycle of greater than $k$. The argument is similar if $t$ is adjacent to a vertex in $P^{2}_{w_{k-2}x_1}$.

\item[$\bullet$] If $v$ has a neighbor in $P^{1}_{a_1a_{k-2}}$. Choose the least indexed vertex in $P^{1}_{a_1a_{k-3}}$, say $c$, such that $\{v,c\}\in E(G)$. 

\begin{itemize}
\item[-] If $v$ does not have a neighbor in $P^{2}_{w_1w_{k-2}}$, then $(t,a_1,\ldots, c, v, x_1, \ldots, x_n=w_{k-2}, \ldots, w_1)$  forms an induced cycle of length greater than $k$.
\item[-] If $v$ has a neighbor in $P^{2}_{w_1w_{k-2}}$, then choose the least indexed vertex in $P^{2}_{w_1w_{k-2}}$, say $y$, such that $v$ is adjacent to $y$. Since $G$ is an $SC_k$ graph, $\mid P^{1}_{a_1c} \mid + \mid P^{2}_{w_1y} \mid = k-2$. Thus, either $(t,w_1,\ldots, y, v, b_1, \ldots, b_{k-n-2}, u)$ or $(t, a_1,\ldots, c, v, x_1, \ldots, x_n, u)$ forms an induced cycle of length greater than $k$ (\emph{see Figure  \ref{fig:is1}(b)}).
\end{itemize}

\end{itemize}

\item[\textbf{Case ($\{t,u\},\{u,v\} \in E(G)$} and \textbf{$\{t,w\} \notin E(G)$):}]
By the definition of $SC_k$, $(P^{1}_{va_{k-2}},u)$ and $(P^{2}_{vw_{k-2}},u)$ forms an induced $C_k$. Thus, $(P^{1}_{ta_{k-2}}, P^{1}_{vb_{k-2}}, P^{2}_{vx_{k-2}}, P^{2}_{tw_{k-2}})$ forms an induced cycle of length greater than $k$. The cycle is induced because the following cases are not possible by the definition of $SC_k$.

\begin{itemize}
\item[$\bullet$] If $t$ is adjacent to some vertices in $P^{1}_{b_1b_{k-2}}$, then choose the largest indexed vertex in $P^{1}_{b_1b_{k-3}}$, say $c$, such that  $\{t,c\}\in E(G)$. If $c=b_{k-3}$, then $(t,c, b_{k-2},u)$ forms an induced $C_4$. If $c \in \{b_1,\ldots,b_{k-4}\}$, then $(P^{1}_{ta_{k-2}}, P^{1}_{cb_{k-2}})$ forms an induced $C_{h > k}$ (\emph{see Figure  \ref{fig:is1}(c)}). Similar argument if $t$ has a neighbor in $P^{2}_{x_1x_{k-2}}$ and if $v$ has an adjacency in $P^{1}_{a_1a_{k-2}}$ or in $P^{2}_{w_1w_{k-2}}$.
\end{itemize}
\end{description}
All the above cases contradict the definition of $SC_k$ graphs. Hence, the lemma is true. $\hfill \qed$
\end{proof}

\begin{definition}
Let $P = \{ P_{u_1u_{l-2}}^{i} \mid 1\leq i \leq n\}$. A graph $G$ is said to be a $cage$ $graph$ of size $n$ denoted as $CAGE(n,l)$ if there exist $w,z \in V(G)$ such that $\{w,u_{1}^{i}\}, \{z,u_{l-2}^{i}\} \in E(G)$ for all $1 \leq i \leq n$ and $P_{u_1u_{l-2}}^{i} $ is a path of length $l-2$. The $CAGE(3,4)$ is shown in \emph{Figure. \ref{fig:pp}}. A $CAGE(n,l)$ is maximum or a maximum cage if there is no $n' > n$ such that $G$ has $CAGE(n',l)$.
\end{definition}

\begin{figure}[h]
\centering
\includegraphics[scale=0.5]{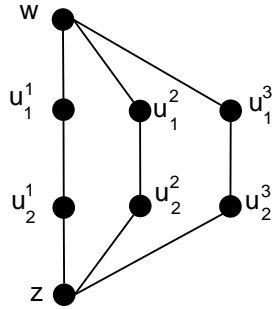}
\caption{$CAGE(3,4)$}
\label{fig:pp}
\end{figure}

\begin{theorem}
\label{mvssc2ksize}
Let $G$ be a connected $SC_k$ graph, $k=2m+4, m \geq 1$. For any two induced cycles $S_i$ and $S_j$ in $G$, if $\vert E(S_i) \cap E(S_j) \vert = k/2$ i.e., $G$ contains $CAGE(3, \frac{k}{2}+1)$, then the cardinality of every minimal vertex separator of $G$ is at most $s$, where $s$ is the size of the maximum cage. 
\end{theorem}
\begin{proof}
On the contrary, assume that there exists a minimal vertex separator $S$ of $G$ such that $\mid S\mid = n$, $n >s$. Since $s \geq 3$, $S$ is an independent set, due to \emph{Lemma \ref{mvsindependent}}. We know that every minimal vertex separator is $(a,b)$-minimal vertex separator for some non-adjacent vertices $a$ and $b$ in $G$. Also, every $(a,b)$-minimal vertex separator is $(c,d)$-minimum vertex separator for some non-adjacent vertices $c$ and $d$ in $G$. Without loss of generality, let us assume that $S$ is a $(c,d)$-minimum vertex separator. Thus, every vertex in $S$ is part of a vertex disjoint path from $c$ to $d$. Hence, we get $CAGE(n,\frac{k}{2}+1)$, where $n > s$. This contradicts the maximality of $s$. Hence the theorem. $\hfill\qed$
\end{proof}

\section{Characterization of $SC_k$ graphs}
Like chordal graphs has a simplicial vertex \cite{tarjan} and chordal bipartite \cite{GolumbicGoss} has a bi-simplicial edge, we shall observe that every $SC_k$ graph has a special vertex or a special $C_k$ namely $pendant$ $vertex$ or \emph{pendant cycle}, respectively. Thus, we can obtain an ordering called \emph{vertex cycle ordering (VCO)} for an $SC_k$ graph. 

\begin{definition}
Let $G$ be an $SC_k, k\geq 5$, graph. A vertex $v \in V(G)$ is said to be a \emph{pendant vertex} if $d_G(v) = 1$. A cycle $S_i$ is said to be $0$-$pendant$ $C_k$ in $G$ if for every cycle $S_j$, $i\neq j$, in $G$, $\vert V(S_i) \cap V(S_j) \vert = 0$ and $S_i$ can have at most one cut vertex of $G$.\\

 A cycle $S_i$ in $G$ is said to be $1$-$pendant$ $C_k$ if $S_i$ has exactly one cut vertex $v$, and there exist at least one induced cycle $S_j$ such that $\vert V(S_i) \cap V(S_j) \vert = 1$ and $S_i$ and $S_j$ shares $v$ in common and with every other cycle $S_m$ in $G$ $\vert V(S_i)\cap V(S_m)\vert = 0$. \\
 
 A cycle $S_i$ in $G$ is said to be $2$-$pendant$ $C_k$ if $S_i$ has exactly one $\{u,v\}$-vertex separator such that $\{u,v\} \in E(G)$ and for all other cycles $S_j$, $S_i$ has vertex intersection with $S_j$ at $u$ or $v$, or edge intersection with $S_j$ at $\{u,v\}$, or no intersection with $S_j$.\\

 A cycle $S_i$ in $G$ is said to be $s$-$pendant$ $C_k$, $s\geq 3$, if there exist at least one cycle $S_j$ in $G$ such that $\vert V(S_i) \cap V(S_j) \vert = s$ and $\vert E(S_i) \cap E(S_j) \vert = s-1$, say $V(S_i) \cap V(S_j) = P_l = (u_1,u_2,\ldots,u_s)$, satisfying the following conditions:
\begin{itemize}
\item[1.] $S_i$ can have $u_1$ or $u_s$ as a cut vertex but not both.
\item[2.] there does not exist a cycle $S_m, m \neq j$ in $G$ such that the graph induced on $ V(S_i)\cap V(S_m)$ is not $P_l$ and $ E(S_i)\cap E(S_m)  \neq \emptyset$.
\end{itemize}   
\end{definition}

\begin{figure}[h]
\centering
\includegraphics[scale=0.3]{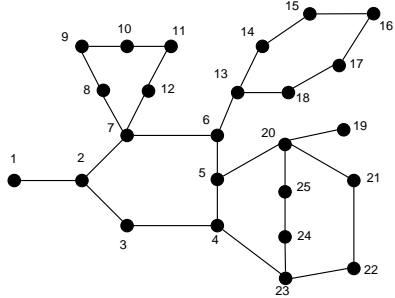}
\caption{An example of an $SC_6$ graph where $\{1\}$ and $\{19\}$ are pendant vertices, $(13,14,15,16,17,18)$ is a $0$-pendant $C_6$, $(7,8,9,10,11,12)$ is a $1$-pendant $C_6$, $(4,5,20,25,24,23)$ is a $2$-pendant cycle and $(20,21,22,23,24,25)$ is a $4$-pendant cycle.}
\end{figure}

\begin{lemma}
\label{specialvertex2k+1}
An $SC_k$ graph $G$ other than $C_k$, $k = 2m+3, m \geq 1$, has any one of the following properties:
\begin{itemize}
\item[(i)] Two non-adjacent pendant vertices
\item[(ii)] Two $s$-pendant $C_k$, $s \in \{0,1,2\}$.
\item[(iii)] An $s$-pendant $C_k$ and a pendant vertex, $s \in \{0,1,2\}$.
\end{itemize}
\end{lemma}
\begin{proof}
 We shall partition the set of $SC_k$ graphs into $SC_k$ graphs with at least one minimal vertex separator of size one and $SC_k$ graphs with every minimal vertex separator is of size two. In both the cases, we shall prove the lemma by mathematical induction on the number of vertices $n$ of $G$.
\begin{description}
\item[Case 1:] There is a minimal vertex separator of size one. \\
\noindent \emph{Base cases:} 
\begin{itemize}
\item[(A)] $G$ be a tree on $n$ vertices, $2 \leq n \leq 2k-1$. Trivially, $G$ has two non-adjacent pendant vertices as there are at least two leaves (degree one vertex) in any tree.
\item[(B)] $G$ is not a tree on $n$ vertices, $n=2k-1$. Clearly, $G$ has two $C_k$ sharing a vertex in common. So, $G$ has two $1$-pendant $C_k$.
\item[(C)] $G$ is a graph different from (A) and (B) on $n$ vertices, $k+1 \leq n \leq 2k-1$. It is easy to see that $G$ has either a $0$-pendant $C_k$ and a pendant vertex or two pendant vertices.
\end{itemize}
\noindent Let $G$ be an $SC_k$ graph with $n \geq 2k$ vertices. Let $S$ be any minimal vertex separator of $G$ such that $\vert S \vert = 1$. Let $G_1$ and $G_2$ be any two connected components in $G \backslash S$. Let $G'$ and $G''$ be the graphs induced on $V(G_1)\cup V(S)$ and  $V(G_2)\cup V(S)$, respectively. If both $G'$ and $G''$ are $C_k$, then there are two $1$-pendant $C_k$'s in $G$. Otherwise, by the induction hypothesis, $G'$ and $G''$ have a $s$-pendant $C_k$, $s \in \{0,1,2\}$, or a pendant vertex, which are also pendant in $G$. Hence the claim.
\item[Case 2:] Every minimal vertex separator is of size two. Let $G$ be an $SC_k$ graph and $S$ be any minimal vertex separator of $G$ such that $S = \{u, v\}$ and $\{u,v\} \in E(G)$.\\
\noindent \emph{Base case:} For $n=2k-2$, an $SC_k$ graph with $2k-1$ edges has two $2$-pendant $C_k$'s.\\
\noindent Let $G$ be an $SC_k$ graph with $n \geq 2k-1$ and $G'$ and $G''$ as defined before. By the hypothesis, $G'$ and $G''$ have a $2$-pendant $C_k$ which are also a $2$-pendant $C_k$ in $G$.
\end{description}
Thus the lemma is true for all $SC_k$ graphs, $k=2m+3, m \geq 1$. $\hfill \qed$

\end{proof}

\begin{lemma}
\label{specialvertex2k}
An $SC_k$ graph $G$ other than $C_k$, $k = 2m+4, m \geq 1$, has any one of the following properties:
\begin{itemize}
\item[(i)] Two non-adjacent pendant vertices.
\item[(ii)] Two $s$-pendant $C_k$, $s \in \{0,1,2,\frac{k}{2}+1\}$.
\item[(iii)] An $s$-pendant $C_k$ and a pendant vertex, $s \in \{0,1,2,\frac{k}{2}+1\}$.
\end{itemize}
\end{lemma}
\begin{proof}
We use induction on $n$, the number of vertices in $G$.\\
\noindent \emph{Base cases:} \vspace{-0.16cm}
\begin{itemize}
\item[(A)] For $2 \leq n \leq 2k-1$, any tree with $n$ vertices has exactly two non-adjacent pendant vertices.
\item[(B)] $G$ is not a tree on $n$ vertices, $n=2k-1$. $G$ has two $1$-pendant $C_k$'s, or two pendant vertices, or a $0$-pendant $C_k$ and a pendant vertex.
\item[(C)] $G$ is not a tree on $\frac{3k}{2}-1$ vertices, $G$ has two $(\frac{k}{2}+1)$-pendant $C_k$, or two pendant vertices, or a $0$-pendant $C_k$ and a pendant vertex.
\item[(D)] $G$ is not a tree on $2k-2$ vertices. $G$ has any one of the following:
\begin{itemize}
\item[$\bullet$] two $2$-pendant $C_k$'s.
\item[$\bullet$] three $(\frac{k}{2}+1)$-pendant $C_k$.
\item[$\bullet$] a pendant vertex and a $2$-pendant $C_k$.
\item[$\bullet$] a pendant vertex and a $(\frac{k}{2}+1)$-pendant $C_k$.
\item[$\bullet$] a $0$-pendant $C_k$ and a pendant vertex.
\item[$\bullet$] two pendant vertices.
\end{itemize}
 
\item[(E)] $G$ is a graph different from (A) and (C) on $n$ vertices, $k+1 \leq n < 2k-2$. $G$ has either a $s$-pendant $C_k$, $s\in \{0,\frac{k}{2}+1\}$, and one pendant vertex or two pendant vertices.
\end{itemize}
\noindent Let $G$ be an $SC_k$ graph with $n \geq 2k$ vertices. Let $S$ be any minimal vertex separator of $G$. Let $G_1$ and $G_2$ be any two connected components in $G \backslash S$. Let $G'$ and $G''$ be the graphs induced on $V(G_1)\cup V(S)$ and  $V(G_2)\cup V(S)$, respectively. If  $ S  = \{u\}$ or $ S  = \{u,v\}$ such that $\{u,v\} \in E(G)$, by the induction hypothesis, both $G'$ and $G''$ have a pendant vertex or a $s$-pendant $C_k$, $s\in\{0,1,2,\frac{k}{2}+1\}$, which are also pendant in $G$. If $\mid S \mid \geq 3$, then $S$ is an independent set, by \emph{Lemma \ref{mvsindependent}}. The possible existence of $(i)$, $(ii)$, or $(iii)$ in this case are as follows:

\begin{itemize}
\item[(1)] If $G'$ (as well as $G''$) has a pendant vertex in $G_1$ (as well as $G_2$), it is also a pendant vertex in $G$.
\item[(2)] If $G'$ (as well as $G''$) has a $s$-pendant $C_k$, $s\in\{0,1,2,\frac{k}{2}+1\}$ in $G_1$ (as well as $G_2$), it is also pendant in $G$.
\item[(3)] If $G'$ and $G''$ do not have any pendant vertices and $s$-pendant $C_k$, $s \in \{0,1,2,\frac{k}{2}+1\}$ in $G_1$ and $G_2$, respectively, and if $G'$ and $G''$ have pendant vertices only in $S$. Since, $S$ is a minimal vertex separator of size greater than two, the only possibility of $G$ is $CAGE(\mid S \mid, \frac{k}{2}+1), l \geq 3$. Thus, $G$ has at least two $(\frac{k}{2}+1)$-pendant $C_k$.
\item[(4)] If either $G'$ has a pendant vertex or a $s$-pendant $C_k$, $s\in \{0,1,2,\frac{k}{2}+1\}$ in $G_1$, and $G''$ has neither of them in $G_2$. If $G'$ itself has any one of $(i)$, $(ii)$ and $(iii)$ in $G_1$, then there is nothing to prove. If $G'$ has a pendant vertex $u$ in $G_1$, then $G\backslash \{u\}$ may have any one of $(i)$, $(ii)$ and $(iii)$, by the hypothesis. If $G\backslash \{u\}$ has none of (i), (ii), and (iii) then, $G\backslash \{u\}$ is a CAGE, thus $G\backslash \{u\}$ has a pendant cycle together with $u$, our claim follows in $G$. If $G\backslash \{u\}$ has $(i)$, then $G$ has two pendant vertices, one is $u$ and the other is from $G\backslash \{u\}$. If $G\backslash \{u\}$ has $(ii)$, then $G$ has a pendant vertex and $s$-pendant cycle from $G\backslash \{u\}$. If $G\backslash \{u\}$ has $(iii)$, then $G$ has a pendant vertex and either a $s$-pendant cycle or a pendant vertex from $G\backslash \{u\}$. If $G'$ has a $s$-pendant $C_k$ in $G_1$, $s\in \{0,1,2,\frac{k}{2}+1\}$, say $C$, then $G\backslash (V(C)\backslash (V(G)\cap V(C)))$ has any one of $(i)$, $(ii)$ and $(iii)$, by the hypothesis. Let $M=G\backslash (V(C)\backslash (V(G)\cap V(C)))$. Thus, if $M$ has $(i)$, then $G$ has a pendant vertex from $M$ and a $s$-pendant $C_k$ from $G_1$, if $M$ has $(ii)$, then $G$ has two $s$-pendant cycle's one from $G_1$ and the other from $M$, if $M$ has $(iii)$, then $G$ has a $s$-pendant $C_k$ from $G_1$ and either a $s$-pendant cycle or a pendant vertex from $M$.
\end{itemize}

Thus the lemma is true for all $SC_k$ graphs, $k=2m+4, m \geq 1$. $\hfill \qed$

\end{proof}

\begin{theorem}
\label{construction}
A connected graph is $SC_k, k\geq 5$, if and only if it can be constructed using the following rules.
\begin{itemize}
\item[(i)] $K_1$ is an $SC_k$ graph.
\item[(ii)] $C_k$ is an $SC_k$ graph.
\item[(iii)] If $G$ is an $SC_k$ graph, then the graph $G'$, where, $V(G') = V(G) \cup \{v\}$, $E(G') = E(G) \cup \{u, v\}$ such that $v \notin V(G)$ and $u$ is any vertex in $V(G)$,  is also an $SC_k$ graph.
\item[(iv)] If $G$ is an $SC_k$ graph, then the graph $G'$, where, $V(G') = V(G) \cup \{v_1, v_2, \ldots, v_{k-1}\}$, $E(G') = E(G) \cup \{\{u, v_1\}, \{v_1, v_2\}, \{v_2, v_3\},  \ldots, \{v_{k-2}, v_{k-1}\}, \{v_{k-1}, u\} \}$ such that $\{v_1, v_2, \ldots, v_{k-1}\} \cap V(G) = \emptyset$ and $u$ is any vertex in $V(G)$,  is also an $SC_k$ graph.
\item[(v)] If $G$ is an $SC_k$ graph, then the graph $G'$, where, $V(G') = V(G) \cup \{v_1, v_2, \ldots, v_{k-2}\}$, $E(G') = E(G) \cup \{\{u, v_1\}, \{v_1, v_2\}, \{v_2, v_3\}, \ldots, \{v_{k-2}, v\} \}$ such that $\{v_1, v_2, \ldots, v_{k-2}\} \cap V(G) = \emptyset$ and $\{u, v\}$ is any edge in $E(G)$,  is also an $SC_k$ graph.
\item[(vi)] If $G$ is an $SC_k$ graph and $ k=2m+4, m \geq 1$, then the graph $G'$, where, $V(G') = V(G) \cup \{v_1, v_2, \ldots, v_{\frac{k}{2}-1}\}$, $E(G') = E(G) \cup \{\{u_1, v_1\}, \{v_1, v_2\}, \{v_2, v_3\}, \ldots, \{v_{\frac{k}{2}-1}, u_{\frac{k}{2}+1}\} \}$ such that $\{v_1, v_2, \ldots, v_{\frac{k}{2}-1}\} \cap V(G) = \emptyset$ and $\{u_1,u_2,\ldots,u_{\frac{k}{2}+1}\}$ is any path of length $\frac{k}{2}+1$ contained in no induced cycle in $G$ or in any one induced cycle $S_i$ of length $k$ in $G$ such that there does not exist an induced cycle $S_j$ in $G$ with $V(S_i) \cap V(S_j) = \{w_1,\ldots,w_{\frac{k}{2}+1}\}$, $w_p = u_p$ for some $p \in\{1,\ldots, \frac{k}{2}+1\}$ and for at least one $q \in\{1,\ldots, \frac{k}{2}+1\}$, $w_q \neq u_q$.
\end{itemize}
\end{theorem} 
\begin{proof}
\begin{description}

\item[{Necessity:}] Given that $G$ is an $SC_k$ graph. By \emph{Lemma \ref{specialvertex2k+1}} and \emph{Lemma \ref{specialvertex2k}}, $SC_k$ graph has at least one pendant vertex or a $s$-pendant $C_k$, $s \in\{0,1,2,\frac{k}{2}+1\}$, and we denote them using the label $x_1$. Consider the graph $G-x_1$ obtained from $G$ by removing the label $x_1$, i.e., remove a pendant vertex or a $s$-pendant $C_k$, $s \in\{0,1,2,\frac{k}{2}+1\}$. Since $SC_k$ graphs respect hereditary property, $G-x_1$ contains a label $x_2$ which is a pendant vertex or a $s$-pendant $C_k$, $s \in\{0,1,2,\frac{k}{2}+1\}$. Repeat the previous step by removing the label $x_2$. Clearly, in at most $n$ iterations we can get an ordering among labels which we call us \emph{vertex cycle ordering(VCO)}. Clearly, the reverse of VCO gives the construction of the underlying $SC_k$ graph. This completes the necessity.
\item[{Sufficiency:}] Let $G'$ be a graph constructed using the rules (i) to (vi). We shall prove the theorem by mathematical induction on the number of iterations needed to construct $G'$.
\begin{description}
\item[Case 1:] $G'$ is obtained by rule (iii).
 
The vertex set and the edge set of the graph $G'$ are $V(G')=V(G) \cup \{v\}$ and $E(G')=E(G) \cup \{x,v\}$, for some $x \in V(G)$, respectively. By the hypothesis, $G$ is an $SC_k$ graph and the newly added edge $\{x, v\}$ does not create any new cycle in $G'$. Thus, $G'$ is also an $SC_k$ graph.

\item[Case 2:] $G'$ is obtained by rule (iv)

For any $u \in V(G)$, $C$ = ($u, u_1, \ldots, u_{k-1}$) be the newly added $C_k$. The vertex set and the edge set of the graph $G'$ are $V(G')=V(G) \cup V(C)$ and $E(G')=E(G) \cup E(C)$, respectively. By the hypothesis, $G$ is an $SC_k$ graph and $C$ does not induce a cycle other than $C_k$ in $G$. Thus, $G'$ is also an $SC_k$ graph.

\item[Case 3:] $G'$ obtained by rule (v)

For any edge $\{u, v\} \in E(G)$, $D$ = ($u, u_1, \ldots, u_{k-2}, v$) be the newly added $C_k$. The vertex set and the edge set of the graph $G'$ are $V(G')=V(G) \cup V(D)$ and $E(G')=E(G) \cup E(D)$, respectively. By the hypothesis, $G$ is an $SC_k$ graph and $D$ does not induce a cycle other than $C_k$ in $G$. Thus, $G'$ is also an $SC_k$ graph.

\item[Case 4:] $G'$ obtained by rule (vi)

For any path of length $\frac{k}{2}+1$, $P=(u_1,u_2,\ldots,u_{\frac{k}{2}+1})$ contained in no induced cycle in $G$. $D=(u_1,v_1,\ldots,v_{\frac{k}{2}-1}, u_{\frac{k}{2}+1},\ldots, u_2)$ be the newly added $C_k$. The vertex set and the edge set of the graph $G'$ are $V(G')=V(G)\cup V(D)$ and $E(G')=E(G)\cup E(D)$, respectively. By the hypothesis, $G$ is an $SC_k$ graph and $D$ does not induce a cycle other than $C_k$ in $G$. Thus, $G'$ is also an $SC_k$ graph. 


For any path of length $\frac{k}{2}+1$, $P=(u_1,u_2,\ldots,u_{\frac{k}{2}+1})$ contained in any one induced cycle $S_i=(u_1,\ldots, u_k)$ of length $k$ in $G$. Let $D=(u_1,v_1,\ldots,v_{\frac{k}{2}-1}, u_{\frac{k}{2}+1},\ldots, u_2)$ be the newly added $C_k$. The vertex set and the edge set of the graph $G'$ are $V(G')=V(G)\cup V(D)$ and $E(G')=E(G)\cup E(D)$, respectively. If there exist an induced $S_j=(w_1,\ldots,w_k)$ in $G$ with $V(S_i) \cap V(S_j) = \{w_1,\ldots,w_{\frac{k}{2}+1}\}$, $w_p = u_p$ for some $p \in\{1,\ldots, \frac{k}{2}+1\}$ and for at least one $q \in\{1,\ldots, \frac{k}{2}+1\}$, $w_q \neq u_q$, then the possible cases are as follows:

\begin{itemize}
\item[$\bullet$] If $(w_1,\ldots,w_{\frac{k}{2}+1}) = (u_2,\ldots,u_{\frac{k}{2}+2})$, then $(u_2 = w_1,u_1,v_1, \ldots, v_{\frac{k}{2}-1}, u_{\frac{k}{2}+1}, u_{\frac{k}{2}+2}, w_{\frac{k}{2}+2}, \ldots, w_{k-1})$ forms an induced cycle of length $k$+$2$.

\item[$\bullet$] If $(w_1,\ldots,w_{\frac{k}{2}+1}) = (u_{\frac{k}{2}},\ldots, u_{k})$, then $(u_1,v_1, \ldots, v_{\frac{k}{2}-1}, w_2=u_{\frac{k}{2}+1}, w_1=u_{\frac{k}{2}}, w_k, w_{k-1}, \ldots,$ $w_{\frac{k}{2}-1}=u_k)$ forms an induced cycle of length $k$+$2$.

\item[$\bullet$] If $(w_1,\ldots,w_{\frac{k}{2}+1}) = (u_s, u_{s+1}, \ldots, u_{s+\frac{k}{2}})$, $s \in \{3,\ldots, \frac{k}{2}-1\}$, then $(u_s, u_{s-1}, \ldots, u_1, v_1, \ldots, v_{\frac{k}{2}-1},$ $u_{\frac{k}{2}+1}, u_{\frac{k}{2}+2}, \ldots, u_{\frac{k}{2}+s}, w_{\frac{k}{2}+2}, \ldots, w_k)$ forms an induced cycle of length $k$-$2$+$s$, which is always greater than $k$.
\end{itemize}

All the above cases contradicts the definition of $SC_k$ graph. Thus, there does not exist an induced $S_j=(w_1,\ldots,w_k)$ in $G$ with $V(S_i) \cap V(S_j) = \{w_1,\ldots,w_{\frac{k}{2}+1}\}$, $w_p = u_p$ for some $p \in\{1,\ldots, \frac{k}{2}+1\}$ and for at least one $q \in\{1,\ldots, \frac{k}{2}+1\}$, $w_q \neq u_q$. By the hypothesis, $G$ is an $SC_k$ graph and $D$ does not induce a cycle other than $C_k$ in $G$. Thus, $G'$ is also an $SC_k$ graph.  $\hfill \qed$
\end{description}
\end{description}
\end{proof}

\begin{lemma}
\label{mindegree}
Let $G$ be an $SC_k$ graph, where $k \geq 5$. Then, the minimum degree of $G$ is at most 2. I.e., $\delta(G) \leq 2$.
\end{lemma}

\begin{proof}
Let us prove the theorem by induction on the length $l$ of VCO of $G$.
\\
\noindent Consider an ordering $(x_1, x_2, \ldots, x_{l+1})$, $l \geq 1$. The label $x_{l+1}$ corresponds to a vertex $v$ or a $C_k$. Let $S$ be the graph corresponds to $x_{l+1}$. Let $M$ be the graph induced on $V(S)\cap V(H)$ in $G$. Thus,
\begin{eqnarray}
\nonumber
\delta(G) &=& \min \{\delta([V(H) \backslash V(M)]), \delta _G(M), \delta([V(S) \backslash V(M)])\} \\ \nonumber
&\leq & \min \{ \delta([V(H) \backslash V(M)]), \delta _G(M), 2\} \\ \nonumber
&\leq & 2 ~(\because \delta(H) \leq 2 ~~and ~~\delta _G(M) \geq 2 ,~ by~ hypothesis) 
\end{eqnarray}
$\hfill \qed$
\end{proof}

\section{Algorithmic results on $SC_k$ graphs}
In this section, we present a polynomial-time algorithm for testing whether an arbitrary graph is an $SC_k$ graph or not, for a fixed $k$. Further, we solve the famous combinatorial problems like coloring, hamiltonicity and treewidth for a given $SC_k$ graph.

\subsection{Recognizing $SC_k$ graphs}
We shall use the ordering on $SC_k$ graphs to test whether the given graph is $SC_k, k\geq 5$, graph or not. First, we present a decomposition theorem for $SC_{k}, k = 2m+4, m \geq 1$, graphs followed by the algorithm for testing $SC_{2k}$ graphs for any fixed $k$. Similarly, we shall produce a decomposition theorem for $SC_{k}, k = 2m+3, m \geq 1$, graphs along with its recognition algorithm.

\begin{definition}
A $bi$-$connected$ $graph$ is a connected graph with no cut vertex. A $bi$-$connected$ $component$ of a graph $G$ is a maximal bi-connected subgraph of $G$. 
\end{definition}

\begin{theorem}
\label{decompositionsc2k}
A graph $G$ is an $SC_{k}$ graph, $k = 2m +4, m \geq 1$, if and only if it can be decomposed into a set of connected components, such that each connected component is any one of the following: 
\begin{itemize}
\item a cut edge
\item a $C_{k}$
\item $CAGE(l, \frac{k}{2}+1)$, $l \geq 3$
\end{itemize}
\end{theorem}
\begin{proof}
\textbf{Necessity:} We shall prove the necessity by mathematical induction on the length $l$ of VCO of $G$.\\
\noindent Consider an ordering $(x_1,\ldots,x_{l+1}), l\geq 1$. The label $x_{l+1}$ corresponds to either a vertex or a $C_k$. 
\begin{itemize}
\item[$\bullet$] 
If $x_{l+1}$ is a vertex $u$, then it is a pendant vertex in $G$ and $\{u\} \cup N_G(u)$ is an edge $e$. Note that $e$ is a cut edge. By the hypothesis, $G\backslash \{u\}$ has a decomposition $D$ where each connected component is a cut edge or a $C_k$ or a $CAGE(l, \frac{k}{2}+1)$, $l \geq 3$. Thus, $G$ can be decomposed into $D$ and a cut edge $e$.

\item[$\bullet$] 
If $x_{l+1}$ is a $s$-pendant $C_k$, $s \in \{0,1,2\}$, say $C$, then by the hypothesis, the graph obtained by the ordering $(x_1,\ldots,x_l)$ has a decomposition $D$ where each connected component is a cut edge or a $C_k$ or a $CAGE(l, \frac{k}{2}+1)$, $l \geq 3$. Thus, $G$ can be decomposed into $D$ and a cycle $C$.

\item[$\bullet$] 
If $x_{l+1}$ is a $(\frac{k}{2}+1)$-pendant $C_k$, say $C$, then by the hypothesis, the graph obtained by the ordering $(x_1,\ldots,x_l)$ has a decomposition $D$ where each connected component is a cut edge or a $C_k$ or a $CAGE(l, \frac{k}{2}+1)$, $l \geq 3$. Now, combine the cycle $C$ to the path $P_l$, which belongs to an induced cycle $C'$ in one of the connected components of $D$ and thus, the corresponding component results in a CAGE, by \emph{Theorem \ref{construction}}. Note that by introducing $x_{l+1}$, either a new CAGE is created or the size of the existing CAGE increased by one. Hence, we obtained a decomposition as per the theorem.

\end{itemize}

\noindent \textbf{Sufficiency:} Given a decomposition of a graph in which every connected component is an $SC_k$ graph. It is clear that, any two connected components are connected either by a vertex or by an edge and this will not induce any new cycle of length, which is not equal to $k$. Hence the claim. $\hfill \qed$
\end{proof}


\noindent From \emph{Theorem \ref{decompositionsc2k}}, we learn that the recognition of $SC_k$ graphs, $k = 2m+4, m \geq 1$, involves two simple steps. Given any arbitrary graph $G$: first, find the decomposition of the graph $G$ such that each connected component is free from the clique separators of size one and two. Now, for each connected component, check whether it is an edge or a 2-regular graph on $k$ vertices or a $CAGE(l, \frac{k}{2}+1)$, $l \geq 3$. If not, $G$ is not an $SC_k$ graph. Note that computing a decomposition where each connected component is free from the clique separators of size one and two for the graph $G$ involves three steps: (1) Find the bi-connected components of $G$, (2) in each component $G_i$ search for an edge $\{u,v\}$ whose removal disconnects $G_i$ (3) if the edge $ \{u,v\}$ exists then decompose $G_i$ as follows: find $G_i\backslash S$, where $S = \{u,v\}$ and add back the edge $\{u,v\}$ to every connected component of $G_i\backslash S$. Do this process recursively in each $G_i$ until there is no component with clique separators of size two. Testing whether a graph is CAGE or not involves the following steps:
\begin{itemize}
\item[1.] Search for two non-adjacent vertices with equal degree and the degree is at least three, say $d$, and all other vertices in the graph should be of degree two. If the above check is unsuccessful, then the given graph is not a CAGE. Otherwise, proceed with the next step. 
\item[2.] Draw BFS tree $T$ rooted at a maximum degree vertex.
\item[3.] To know whether $T$ corresponds to $CAGE(d, \frac{k}{2}+1)$, check whether the number of levels in $T$ is $\frac{k}{2}+1$, the root has degree $l$, and there are $(d-1)$ slanting edges between the last two levels. Further, the last level has exactly one vertex and $(d-1)$ slanting edges are from $v$ to all other vertices at last but one level except its parent. 
\end{itemize}

Clearly, all the above steps can be verified using the standard BFS and hence test can be done in $O(n+m)$ time, where $n$ and $m$ denotes the number of vertices and edges in $G$, respectively.

\begin{theorem}
\label{decompositionsc2k+1}
A graph $G$ is an $SC_{k}, k = 2m+3, m \geq 1$, graph if and only if it can be decomposed into a set of connected components, where every connected component of $G$ is any one of the following:
\begin{itemize}
\item[(i)] a cut edge
\item[(ii)] a $C_k$
\end{itemize}
\end{theorem}
\begin{proof}
\textbf{Necessity:}  We shall prove this by mathematical induction on the length $l$ of VCO of $G$.
\noindent Consider an ordering $(x_1,\ldots,x_{l+1}), l\geq 1$. The label $x_{l+1}$ corresponds to either a vertex or a $C_k$. 
\begin{itemize}
\item[$\bullet$] 
If $x_{l+1}$ is a vertex $u$, then it is a pendant vertex in $G$ and $\{u\} \cup N_G(u)$ is a cut edge $e$ in $G$. By the hypothesis, $G\backslash \{u\}$ has a decomposition $D$ where each connected component is an edge or a $C_k$. Thus, $G$ can be decomposed into $D$ and a cut edge $e$.

\item[$\bullet$] 
If $x_{l+1}$ is a $s$-pendant $C_k$, $s \in \{0,1,2\}$, say $C$, then by the hypothesis, the graph obtained by the ordering $(x_1,\ldots,x_l)$ has a decomposition $D$ where each connected component is a cut edge or a $C_k$. Thus, $G$ can be decomposed into $D$ and a cycle $C$.

\end{itemize}

\noindent \textbf{Sufficiency:} Given a decomposition of a graph in which every connected component is an $SC_k$ graph. It is clear that, any two connected components are connected either by a vertex or by an edge and this will not induce any new cycle of length, which is not equal to $k$. Hence the claim. $\hfill \qed$
\end{proof}


\noindent From \emph{Theorem \ref{decompositionsc2k+1}}, we observe that the recognition of $SC_k$ graphs, $k = 2m+3, m \geq 1$, involves two simple steps. Given any arbitrary graph $G$: first, find the decomposition of the graph $G$ such that each connected component is free from the clique separators of size one and two. Now, for each connected component, check whether it is an edge or a 2-regular graph on $k$ vertices. If not, $G$ is not an $SC_k$ graph. Thus, we can recognize $SC_k$ graphs, $k = 2m+3, m \geq 1$, using BFS in $O(n+m)$ time, where $n$ and $m$ denotes the number of vertices and edges in $G$, respectively.

\subsection{Structure of non-tree edges in $SC_k$ graphs}
\begin{definition}
Let $G$ be a connected graph and $T$ be the \emph{Breadth First Search} ($BFS$) tree of $G$. Let $E(G)$ denotes the edges in the graph $G$ and $E(T)$ denotes the edges in the BFS tree $T$. The \emph{non-tree edges} are the edges in $E(G)\backslash E(T)$ i.e., the edges which are in graph $G$ but not in tree $T$.  
\end{definition}

\begin{definition}
Let $G$ be a connected graph and $T$ be the \emph{Breadth First Search} ($BFS$) tree of $G$. The set $E(G)\backslash E(T)$ is called as non-tree edges. A non-tree edge, $\{u,v\} \in E(G)\backslash E(T)$ is said to be a \emph{cross edge} if both $u$ and $v$ are in same levels of the tree $T$. A non-tree edge, $\{u,v\} \in E(G)\backslash E(T)$ is said to be a \emph{slanting edge} if both $u$ and $v$ are in adjacent levels of the tree $T$.
\end{definition}

\begin{definition}
A \emph{matching} in a graph $G$ is a set of independent edges.
\end{definition}

\begin{lemma}
\label{matching}
 Let $T$ be the BFS tree of an $SC_{2k+1}, k \geq 2$, graph $G$, then the set of non-tree edges of $G$ forms a matching.
\end{lemma}
\begin{proof}
Construct a BFS tree $T$ for the given graph $G$ by fixing $r$ as a root. Since, $G$ is an $SC_{2k+1}, k \geq 2$ graph, the case where every non-tree edge in $T$ is a slanting edge, is not possible. Now our claim is to prove that the non-tree edges of T forms a matching. On the contrary, assume that the non-tree edges of T do not form a matching. i.e., there exist at least two non-tree edges with a common vertex. We shall partition the $SC_{2k+1}$ graphs into the graphs which has only cross edges in $T$ and the graphs which has both cross edges and slanting edges in $T$.

\begin{description}
\item[\bf{Case 1:}] The only non-tree edges in $T$ are cross edges. By our assumption,  there exist cross edges $e, f \in E(G)\backslash E(T)$ in the least level $l$ such that $e = \{u, v\}$ and $f = \{v, w\}$.

\begin{figure}[h]
\begin{center}
\includegraphics[scale=0.25]{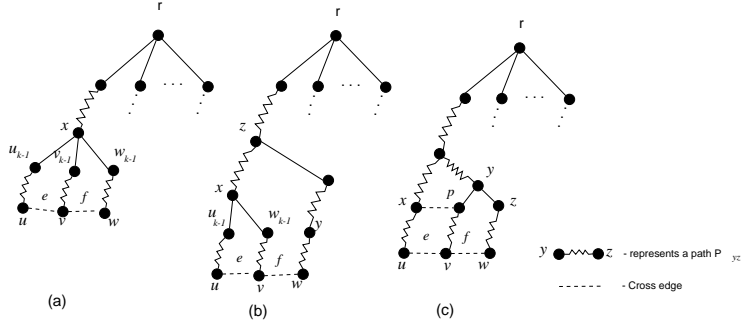}
\caption{BFS Tree $T$ of $G$ with cross edges $e$ and $f$}
\label{bfs1}
\end{center}
 \vspace{-0.7cm}
\end{figure}

\begin{description}
\item[$\bullet$] If for some $x \in V(T)$, $(P_{ux}, P_{vx})$ and $(P_{vx}, P_{wx})$ forms an induced $C_{2k+1}$ in $T$, where $\{u,v\}$ and $\{v,w\}$ are cross edges and all other edges are in $E(T)$, then $(P_{ux}, P_{wx}, v)$ forms an induced $C_{2k+2}$ (\emph{see} \emph{Figure }\ref{bfs1}(a)). 

\item[$\bullet$] If for some $x \in V(T)$, $(P_{ux}, P_{vx})$ forms an induced $C_{2k+1}$ and if there exists $y, z \in V(T)$ and $P_{wy}$ in $T$ such that $z$ is a common parent of $x$ and $y$, then $(P_{vx}, P_{xz}, P_{zy}, P_{yw})$ forms an induced $C_n$, $n > 2k+1$  (\emph{see} \emph{Figure }\ref{bfs1}(b)). 

\item[$\bullet$] If for some $y\in V(T)$, $(P_{yw},P_{vy})$ forms an induced $C_{2k+1}$ in $T$, where $\{v,w\}$ is a cross edge, and for some $x \in V(T)$ and $p \in V(P_{yv})$, $\{x,p\}$ is a cross edge, then $(P_{xu},P_{vp})$ forms an induced cycle of even length (\emph{see} \emph{Figure }\ref{bfs1}(c)).
\end{description}

\item[\bf{Case 2:}] The non-tree edges in $T$ contains both cross edges and slanting edges. By our assumption, there exist an edge $e =\{u,v\} \in E(G)\backslash E(T)$ in level $l$ and an edge $f = \{v,w\} \in E(G)\backslash E(T)$ where $w$ is in level $l+1$ or in level $l-1$ such that $l$ is the least possible level.

\begin{figure}[h]
\begin{center}
\includegraphics[scale=0.25]{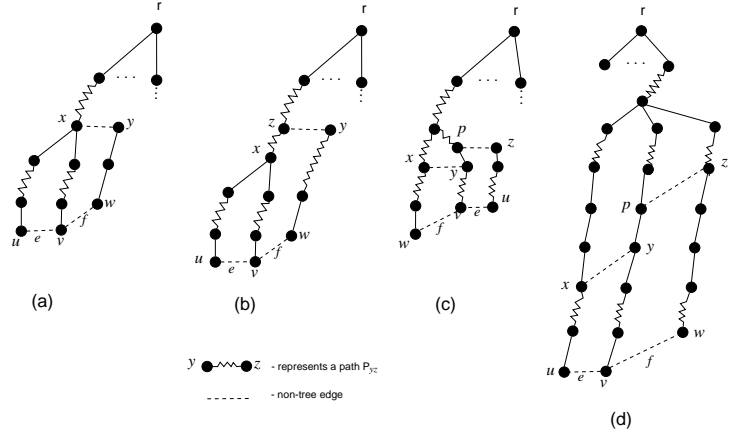}
\caption{BFS Tree $T$ of $G$ with cross edge $e$ and slanting edge $f$}
\label{bfs2}
\end{center}
\end{figure}

\begin{description}

\item[$\bullet$] If $w$ is in level $l-1$ and if for some $x, y \in V(T)$, $(P_{ux}, P_{vx})$ and $(P_{vx}, P_{wy})$ forms an induced $C_{2k+1}$ in $T$, where $\{u,v\}$ and $\{x,y\}$ are cross edges, $\{v,w\}$ is a slanting edge and all other edges are in $E(T)$, then $(P_{ux}, P_{wy}, v)$ forms an induced $C_{2k+2}$ (\emph{see} \emph{Figure }\ref{bfs2}(a)).

\item[$\bullet$] If $w$ is in level $l-1$ and for some $x, y, z \in V(T)$, $(P_{ux}, P_{vx})$ forms an induced $C_{2k+1}$ in $T$, $P_{wy}$ exists in $T$ and $z$ is the common parent of $x$ and $y$ where $\{u,v\}$ and $\{z,y\}$ are cross edges, $\{v,w\}$ is a slanting edge and all other edges are in $E(T)$, then $(P_{vx}, P_{xz}, P_{wy})$forms an induced $C_{2k+2}$ (\emph{see} \emph{Figure }\ref{bfs2}(b)).

\item[$\bullet$] If for some $x, y\in V(T)$, $(P_{xw},P_{vy})$ forms an induced $C_{2k+1}$ in $T$, where $\{v,w\}$ is a slanting edge and $\{x,y\}$ is a cross edge, and for some $p,z \in V(T)$, $\{z,p\}$ forms a cross edge, then $(p, P_{yv},P_{uz})$ forms an induced cycle of even length (\emph{see} \emph{Figure }\ref{bfs2}(c)). 

\item[$\bullet$] If for some $x, y\in V(T)$, $(P_{xu},P_{vy})$ forms an induced $C_{2k+1}$ in $T$, where $\{u,v\}$ is a cross edge and $\{x,y\}$ is a slanting edge, and for some $p,z \in V(T)$, $\{z,p\}$ forms a slanting edge and $\{p,y\} \in E(T)$, then $(p, P_{yv},P_{wz})$ forms an induced cycle of even length (\emph{see} \emph{Figure }\ref{bfs2}(d)). 


\end{description}

\item[\textbf{Case 3:}] The non-tree edges in $T$ contains both cross edges and slanting edges. By our assumption, there exists two slanting edges $e =\{u,v\} \in E(G)\backslash E(T)$  and $f = \{v,w\} \in E(G)\backslash E(T)$ where $u$ is in level $l-1$, $v$ is in level $l$ and $w$ is in level $l+1$ such that $l$ is the least possible level. 

\begin{description}
\item[$\bullet$] If $(P_{wx},P_{vy})$ and $(P_{pv},P_{zu})$ forms an induced $C_{2k+1}$, where $\{x,y\}$ and $\{p,z\}$ are cross edges, then $(z,p,y,x,P_{xw},v,u,P_{zu})$ forms an induced $C_{2k+4}$ 
(\emph{see} \emph{Figure }\ref{fig:bfs3}(a)).    

\item[$\bullet$] If for some $x,y,p,z \in V(T)$, $\{x,y\}$ is a slanting edge and $\{p,z\}$ is a cross edge such that there exist $P_{xw}$, $P_{yv}$, $P_{zu}$ and $\{y,p\}\in E(T)$. Thus, $(P_{pv},P_{zu})$ forms an induced $C_{2k+1}$ and $(P_{xw}, P_{yv})$ forms an induced cycle of even length 
(\emph{see} \emph{Figure }\ref{fig:bfs3}(b)).    
\end{description}

\begin{figure}[h]
\begin{center}
\includegraphics[scale=0.25]{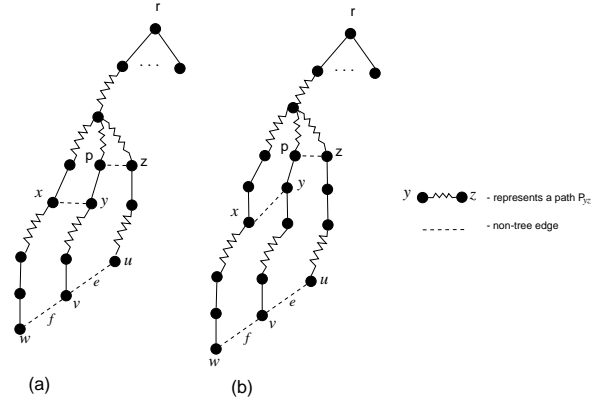}
\caption{BFS Tree $T$ of $G$ with two slanting edges $e$ and $f$}
\label{fig:bfs3}
\end{center}
\end{figure}

\end{description}
All the above cases contradicts the definition of $SC_{2k+1}$ graphs. Hence our assumption, cross edges does not form a matching is wrong. Thus, cross edges in $T$ forms a matching. 
$\hfill \qed$
\end{proof}


\subsection{Hamiltonicity in $SC_k$ graphs}

In this subsection, we provide a necessary and sufficient condition for the existence of hamiltonian cycle in $SC_k$ graphs.

\begin{definition}
 An $SC_k$ graph is said to be \emph{$n$-$C_k$ pyramid} if it has $(k-2)n+2$ vertices, $(k-1)n+1$ edges, exactly two adjacent vertices of degree $n+1$ and every other vertices are of degree two. A 3-$C_5$ pyramid is shown in Figure  \ref{fig: 3c5}.
\end{definition}

\begin{figure}[h]
\begin{center}
\includegraphics[scale=0.30]{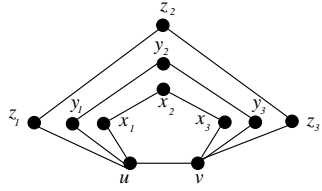}
\vspace{-0.5cm}
\caption{3-$C_5$ Pyramid.}
\label{fig: 3c5}
\end{center}
\end{figure}

\begin{definition}
 The graph $G$ is \emph{Hamiltonian} if it has a spanning cycle (a cycle that contains all vertices in $G$), also called a Hamiltonian cycle.
 \end{definition}
 
\begin{theorem}
\emph{(Chvatal \cite{chvatal})} If a connected graph $G$ has a Hamiltonian cycle, then for each $S \subset V(G)$, the graph $G\backslash S$ has at most $\vert S \vert $ components.
\end{theorem}

\begin{lemma}
\label{ncknonhamiltonian}
A $n$-$C_k$ pyramid graph is non-hamiltonian for all $n \geq 3$ and $k \geq 5$.
\end{lemma}
\begin{proof}
Let $G$ be a $n$-$C_k$ pyramid graph, $n \geq 3$ and $k \geq 5$. Let $u$ and $v$ be the two adjacent vertices of degree $n+1$ in $G$. Let $S = \{u, v\}$. By Chvatal's theorem, $G$ is not a Hamiltonian graph, as the graph $G\backslash S$ will disconnect the graph into $n$ connected components. Hence the lemma.
$\hfill \qed$
\end{proof}

\begin{lemma}
\label{ncknonhamiltoniansubclass}
Any $SC_k$-graph $G$ which contains $n$-$C_k$ pyramid, $k \geq 5$, as an induced subgraph is non-hamiltonian.
\end{lemma}
\begin{proof}
On the contrary, assume that $G$ is hamiltonian. Let $u$ and $v$ be the adjacent vertices of degree greater than or equal to $n+1$ in $G$ such that $\{u, v\}$ is an edge of $n$-$C_k$ pyramid. Let $S_1, S_2, \ldots, S_n$ be the $n$ cycles containing the edge $\{u, v\}$ in $G$. Let $S_1 = (x_1, x_2, \ldots, x_{k-2}, u, v)$ and $S_2 = (y_1, y_2, \ldots, y_{k-2}, u, v)$ (\emph{for e.g.,} \emph{see} \emph{Figure } \ref{fig: 3c5}, where $n=3$ and $k=5$). Since $G$ is hamiltonian, there exist a path from $S_i$ to $S_j$, $i \neq j$ other than the path through the edge $\{u, v\}$. In particular, there exist a path from $S_1$ to $S_2$ which does not pass through the vertices $\{u, v\}$. i.e., there exist at least one path $P$ from $x_i$ to $y_j$, $1 \leq i, j \leq k-2$ which does not pass through the vertices $\{u, v\}$, which contradicts the construction of $SC_k$ graph. Hence our assumption is wrong, which implies $G$ is non-hamiltonian. $\hfill \qed$ 
\end{proof}

\begin{theorem}
 Let $G$ be an $SC_k$ graph. $G$ is Hamiltonian if and only if it is 2-connected, $CAGE(\frac{k}{2}+1, 3)$ free and 3-$C_k$ pyramid free.
\end{theorem}
\begin{proof}
\emph{Necessity:} we know that every hamiltonian graph is $2$-connected, thus, $G$ is $2$-connected. Now our claim is to prove $G$ is 3-$C_k$ pyramid free. On the contrary, assume that $G$ has 3-$C_k$ pyramid as an induced subgraph. Thus, there exist an edge $\{u,v\} \in E(G)$ such that $G$ has at least three $C_k$'s, say $S_1, S_2$, $S_3$, with the property $\bigcap_{i = 1}^{3} E(S_i) = \{\{u, v\}\}$. By \emph{Lemma \ref{ncknonhamiltonian}} and \emph{Lemma \ref{ncknonhamiltoniansubclass}}, $G$ is non-hamiltonian, which is a contradiction. Therefore, $G$ is 3-$C_k$ pyramid free. Also, by the definition of CAGE it is clear that $CAGE( \frac{k}{2}+1, 3)$ is non-hamiltonian, thus, $G$ is $CAGE(\frac{k}{2}+1,3)$ free.
\\
\emph{Sufficiency:} Let $S = \{e \in E(S_i)\cap E(S_j) \mid S_i\mbox{ and }S_j \text{ are induced cycles in }G \}$. Consider a graph $H$, where $V(H) = V(G)$ and $E(H) = E(G)\backslash S$. Since, $G$ is $2$-connected, by \emph{Theorem \ref{construction}}, the graph $G$ is constructed only by \emph{rule (ii)} and \emph{rule (v)}. Therefore, the graph $H$ is an induced cycle, which is a spanning cycle in $G$. Hence, $G$ is hamiltonian. $\hfill \qed$
\end{proof}

\section{Treewidth of $SC_k$ graphs}
 A \emph{tree decomposition}\cite{kloks} of a graph $G$ is a pair ($T,X$) where $T$ is a tree and $X$ assigns a set $X_t \subset V(G)$ to each vertex $t$ of $T$ such that 
\begin{itemize}
\item[(i)] $V(G) = \bigcup_{t \in V(T)} X_t$,
\item[(ii)] for every edge $\{u,v\} \in E(G)$, there is some $t \in V(T)$ such that $u, v \in X_t$ and
\item[(iii)] for every vertex $u \in V(G)$, the set $\{ t \in V(T)\vert u \in X_t \}$ induces a subtree of the tree $T$.
\end{itemize}
The \emph{width} of a tree decomposition ($T,X$) is $\max_{t \in V(T)} \vert X_t \vert - 1$ and the \emph{tree-width}, $tw(G)$, of $G$ is the minimum width of all tree decompositions of $G$.  

\begin{definition}
A graph is a $k$-$tree$ if every minimal vertex separator of $G$ is of size $k$ and every maximal clique is of size $k+1$. A graph $G$ is said to be a \emph{partial-k-tree} if it is an edge subgraph of a $k$-tree.
\end{definition}

In this section, we present an exact bound for the treewidth followed by an algorithm which gives a tree decompositon ($T, X$) with $\max_{t \in V(T)} \vert X_t \vert = 2$ or $3$ for the given $SC_k, k \geq 5$, graph. Let $G$ be an $SC_k, k \geq 5$ graph. We know that $tw(G) \geq \omega (G) - 1$. Since, $K_2$ is the maximum clique in $G$, $tw(G) \geq 1$. We can divide $SC_k$ graphs into $SC_k$ graphs with cycles and $SC_k$ graphs without cycles. It is clear that, $SC_k$ graphs without cycles are same as trees and we know that $tw(tree) = 1$, i.e., $\max_{t \in V(T)} \vert X_t \vert = 2$. Thus, in this section, we consider $SC_k$ graphs with cycles. It is evident that the lower bound of $SC_k$ graphs is two as $tw(C_n)=2$. We observe that the upper bound for $SC_k$ graphs is two by proving that $SC_k$ graphs are partial-2-trees, an edge subgraph of a 2-tree. Alternatively, we augment edges to the given $SC_k$ to produce a 2-tree and the augmentation algorithm is given below.

\begin{definition}
A \emph{minimum fill-in} of a graph $G$ is the minimum number of edges whose addition makes the graph $G$ chordal.
\end{definition}

\begin{algorithm}
\caption{$\mathtt{Fill-in~~ of~~ SC_k ~~graph}$}
\begin{algorithmic}[1]
\STATE{\textbf{INPUT:} An $SC_k$ graph, $G$.}
\STATE{\textbf{OUTPUT:} A chordal graph $G'$}
\STATE{Decompose the graph $G$ into a set of connected components as per \emph{Theorems \ref{decompositionsc2k+1} and \ref{decompositionsc2k}}.}
\STATE{Let $G_1, G_2, \ldots, G_l$ be the connected components of the decomposed graph.}
\FOR{$i=1$ to $l$}
		\IF{$G_i$ is a $C_k$}
			\STATE{Choose any one vertex $u$ in $G_i$ and make it adjacent to all the non-adjacent vertices of $C_k$ in $G_i$.}
		\ELSIF{$G_i$ is a CAGE}
			\STATE{Choose a vertex with maximum degree and make it adjacent to all the non-adjacent vertices in the CAGE.}
		\ENDIF
\ENDFOR
\STATE{Now combine the decomposed graph into a graph $G'$ and Return $G'$.}
\end{algorithmic}
\label{alg:partial2tree}
\end{algorithm}

\begin{theorem}
\label{fill-in}
The algorithm $\mathtt{Fill-in( )}$ outputs a chordal graph, which is a partial-2-tree.
\end{theorem}
\begin{proof}
We prove this by induction on the length of the VCO of a given $SC_k$ graph.
 In the ordering $(x_1, \ldots, x_n)$, let $G$ be an $SC_k$ graph obtained after $n^{th}$ ordering, $x_n, n \geq 2$. Our claim is to prove $G$ is a chordal graph and a partial-2-tree.
\begin{description}
\item[\textbf{Case 1:}] $x_n$ is $P_1$, say $u$.\\
By the hypothesis, it is clear that $G$ is chordal and a partial-2-tree.

\item[\textbf{Case 2:}] $x_n$ is a $0$-pendant $C_k$ or a $1$-pendant $C_k$.\\
Let $x_n$ be $C=(v_1,\ldots, v_k)$ and $H$ be the associated graph for the ordering $(x_1,\ldots,x_{n-1})$. W.l.o.g, $v_1\in V(H)$. By the induction hypothesis, when $H$ is passed as an input to the Algorithm \ref{alg:partial2tree}, the output of Algorithm \ref{alg:partial2tree} is a chordal graph and a partial-2-tree. Now \emph{Step 7} of Algorithm \ref{alg:partial2tree} adds edges from $v_1$ to all the non-adjacent vertices of $C$. Clearly, the resulting graph is chordal and a partial-2-tree.

\item[\textbf{Case 3:}] $x_n$ is a $2$-pendant $C_k$.\\
Let $x_n$ be $C=(v_1,\ldots, v_k)$. Let $H$ be the associated graph for the ordering $(x_1,\ldots,x_{n-1})$ and by the induction hypothesis, when $H$ is given as an input to the Algorithm \ref{alg:partial2tree}, the output of Algorithm \ref{alg:partial2tree} is a chordal graph and a partial-2-tree. Since, $x_n$ is $2$-pendant vertex, w.l.o.g, let $\{v_1,v_2\}\in E(H)$. Now, augment edges from $v_1$ to every non-adjacent vertex of $C$. Clearly, the resulting graph is chordal and a partial-2-tree.

\item[\textbf{Case 4:}] $x_n$ is a $\left(\frac{k}{2}+1\right)$-pendant $C_k$.\\
Let $x_n$ be $C=(v_1,\ldots, v_k)$ and $H$ be the associated graph for the ordering $(x_1,\ldots,x_{n-1})$. By the induction hypothesis, when $H$ is passed as an input to the Algorithm \ref{alg:partial2tree}, the output of Algorithm \ref{alg:partial2tree} is a chordal graph and a partial-2-tree. W.l.o.g, assume that $ \{v_1,\ldots, v_{\frac{k}{2}+1}\}\subset V(H)$. Now \emph{Step 9} of Algorithm \ref{alg:partial2tree} adds edges from $v_1$ to all the non-adjacent vertices of $C$. Clearly, the resulting graph is chordal and partial-2-tree. $\hfill \qed$
\end{description}
\end{proof}

\noindent From the above case analysis, it follows that $tw(G) \leq 2$. Since $tw(G) \geq 2$, $tw(G) = 2$.

\begin{corollary}
\label{minfillin}
Minimum fill-in of $SC_k$ graphs is polynomial-time solvable.
\end{corollary}
\begin{proof}
The output of Algorithm \ref{alg:partial2tree} yields a chordal graph by augmenting a minimum number of edges. Therefore, the output is precisely the minimum fill-in of $SC_k$ graphs. Further, minimum fill-in is polynomial-time solvable for $SC_k$ graphs. Note that the number of edges augmented in a given graph $G$ by Algorithm \ref{alg:partial2tree} is $a(k-3)+(n_1+n_2+\ldots +n_s)(\frac{k}{2}-2)+1)$, where $a$ denotes the number of $C_k$'s and $n_i$ denotes the $CAGE(n_i,\frac{k}{2}+1)$ in the decomposition of $G$.
$\hfill \qed$
\end{proof}

\noindent Having given the bounds for treewidth, we now present an algorithm which gives a tree decomposition for $SC_k, k =2m+3, m \geq 1$ graphs, where $\max_{t \in V(T)} \vert X_t \vert = 3$.\\

\noindent \textbf{Outline of the algorithm:}
The algorithm first constructs a graph $G'$ from $G$ as follows: to start with, every induced cycle in $G$ is converted into a collection of $P_3$'s appropriately, where the weights of the edges are assigned to be one. Next, the algorithm collects all the edges in $S$ which are not a part of any cycle in $G$. Now, for every element in $S$, the algorithm creates a new vertex. Finally, the algorithm adds weighted edges among the newly formed vertices and the newly constructed $P_3$'s, and the weights of the edges depends on its end vertices.  Thus, the graph $G'$ has been constructed from $G$. Now, find the minimum spanning tree $T$ for the weighted graph $G'$ and the algorithm outputs $T$ as a tree decomposition for $G$.

\begin{algorithm}
\caption{Tree Decomposition for $SC_{k}, k= 2m+3, m \geq 1$, graphs}
\begin{algorithmic}[1]
\STATE{\textbf{Input:} $SC_k$ graph $G$ with cycles, $k= 2m+3, m \geq 1$.}
\STATE{\textbf{Output:} Tree decomposition of $G$}
\STATE{Let $\{v_1, \ldots, v_n\}$ be the vertex set of $G$ and $S_1, S_2, \ldots, S_l$, $l \geq 0$ be the cycles in $G$.}
\FOR{$i=1$ to $l$}
\STATE{for every cycle $S_i = (x_{1}^{i}, x_{2}^{i}, x_{3}^{i}, \ldots, x_{k-1}^{i}, x_{k}^{i})$ in $G$ define $P^{i1}, P^{i2}, P^{i3}, \ldots, P^{i(k-2)}$ as follows.
\begin{itemize}
\item[$\bullet$] $P^{i1} = \{ x_{1}^{i}, x_{2}^{i}, x_{3}^{i} \mid$ $(x_{1}^{i}, x_{2}^{i}, x_{3}^{i})$ is the induced $P_3$ in $G$\}, 
\item[$\bullet$] $P^{i2} = \{ x_{1}^{i}, x_{3}^{i}, x_{4}^{i}\}$,
\item[$\bullet$] $P^{i3} = \{x_{1}^{i}, x_{4}^{i}, x_{5}^{i}\}$, $\ldots$, $P^{i(k-2)} = \{x_{1}^{i}, x_{k-1}^{i}, x_{k}^{i}\}$.
\end{itemize}}
\ENDFOR
\STATE{Let $S=\{e=\{x,y\} \mid$ $x \notin V(S_i), 1 \leq i \leq l$ or $y \notin V(S_i), 1 \leq i \leq l \}$}
\STATE{Let $e_1, e_2, \ldots , e_s$ ($s \geq 0$) be the edges in $S$.}
\FOR{$i=1$ to $s$}
\STATE{$X_i = \{x, y \mid e_i = \{x, y\}\} $}
\ENDFOR
\STATE{$p=s+1$}
\FOR{$i=1$ to $l$}
\FOR{$j=1$ to $k-2$}
\STATE{$X_p = P^{ij}$ and $p=p+1$}
\ENDFOR
\ENDFOR
\STATE{For each $X_i$, $s+1\leq i \leq s+(k-2)l$, relabel the vertices in $X_i$ by the labels given for the vertices during the input.}
\STATE{Construct a graph $G'$ with vertex set, $V(G') = \{X_i \mid 1 \leq i \leq s+(k-2)l\}$ and two vertices $X_i, X_j \in V(G')$, $i \neq j$ are adjacent if any one of the following types is true:
\begin{itemize}
\item[Type 1:] $\vert X_i \vert = \vert X_j \vert = 3$ and $X_i, X_j \in S_m$, for some $1 \leq m \leq l$, then, $\vert V(X_i) \cap V(X_j) \vert = 2$. 
\item[Type 2:] $\vert X_i \vert = \vert X_j \vert = 3$ and $X_i \in S_r, X_j \in S_m$, for some $1 \leq m, r \leq l$ and $m \neq r$ and $\vert E(S_m) \cap E(S_r) \vert = 1$, then, $\vert V(X_i) \cap V(X_j) \vert = 2$.
\item[Type 3:] $\vert X_i \vert = \vert X_j \vert = 3$ and $X_i \in S_r, X_j \in S_m$, for some $1 \leq m, r \leq l$ and $m \neq r$ and $\vert V(S_m) \cap V(S_r) \vert = 1$, then, $\vert V(X_i) \cap V(X_j) \vert = 1$.
\item[Type 4:] $\vert X_i \vert = 3$ and $\vert X_j \vert = 2$ and $\vert V(X_i) \cap V(X_j) \vert = 1$.
\item[Type 5:] $\vert X_i \vert = 2$ and $\vert X_j \vert = 2$ and $\vert V(X_i) \cap V(X_j) \vert = 1$.
\end{itemize}}
\STATE{Convert the unweighted graph $G'$ to a weighted graph $G''$ by assigning the weight $i$ for the edges of type $i$, $1 \leq i \leq 5$. }
\STATE{Find a minimum spanning tree $T$ for the weighted graph $G''$}
\STATE{\emph{Return $T$}}
\end{algorithmic}
\label{treedecompositionsc2k+1}
\end{algorithm}

\vspace{0.4cm}

\noindent \textbf{{Trace of the algorithm}}
\vspace{0.3cm}

\noindent We trace the steps of the Algorithm \ref{treedecompositionsc2k+1} in \emph{Figure \ref{fig:tracetreedecompositionsc2k+1}}.

\begin{figure}[H]
\centering
\includegraphics[scale=0.26]{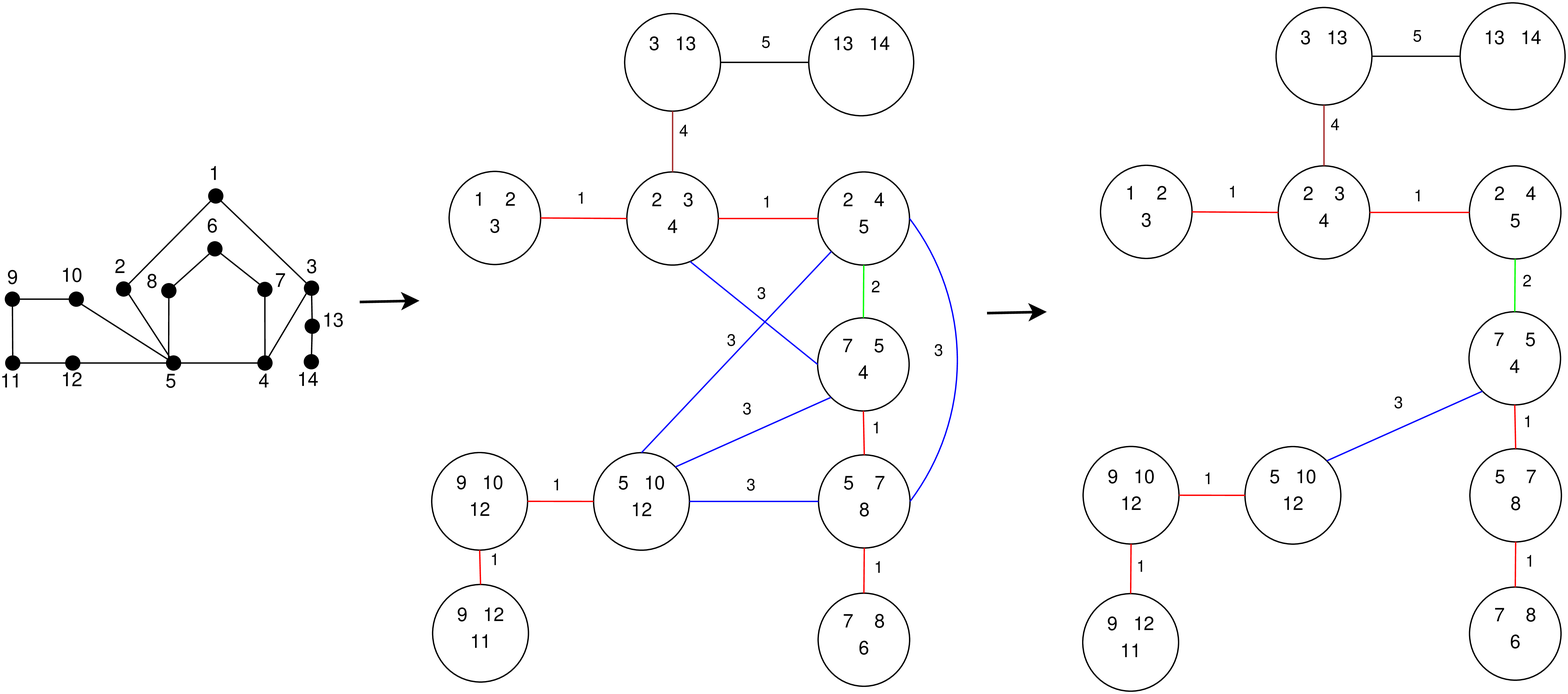} 
\vspace{-0.1cm}
\caption{Tree decomposition of an $SC_5$ graph.}
\label{fig:tracetreedecompositionsc2k+1}
\end{figure}

\begin{itemize}
\item[1.] Input is an $SC_5$ graph $G$. For cycles $S_1=(4,5,8,6,7)$, $S_2=(4,5,2,1,3)$ and $S_3=(9,10,5,12,11)$ in the graph $G$, create $P^{11} = \{7,5,4\}$, $P^{12} = \{5,7,8\}$, $P^{13} = \{6,7,8\}$, $P^{21} = \{1,2,3\}$, $P^{22} = \{2,3,4\}$, $P^{23} = \{2,4,5\}$, $P^{31} = \{9,11,12\}$, $P^{32} = \{9,10,12\}$ and $P^{33} = \{5,10,12\}$.
\item[2.] Now assign, $X_1 = \{3,13\}$, $X_2 = \{13,14\}$, $X_3 = P^{11}$, $X_4 = P^{12}$, $X_5 = P^{13}$, $X_6 = P^{21}$, $X_7 = P^{22}$, $X_8 = P^{23}$, $X_9 = P^{31}$, $X_{10} = P^{32}$ and $X_{11} = P^{33}$.
\item[3.] Draw edges between $X_i$ and $X_j$, $i\neq j$ if it obeys the \emph{line 19} and assign weights for edges as in \emph{lines 19-20} (\emph{see} \emph{Figure 1}).
\item[4.] Construct a minimum weight spanning tree, $T'$, for the graph $T$. Thus, the algorithm is complete and results a tree decomposition with minimum tree width for the given graph $G$.
\end{itemize}

\begin{theorem}\label{prooftreedecomp2k+1}
The graph $T$ obtained from the \emph{Algorithm \ref{treedecompositionsc2k+1}} is a tree decomposition of $G$ such that $tw(G) = 2$.
\end{theorem}
\begin{proof}
 Our claim is to prove that the graph $T$ is a tree and all the three conditions of tree decomposition are satisfied by $T$.
\begin{description}
\item[claim 1:] $T$ is a tree \\
It is clear from the construction of the graph $G'$, that the graph $G'$ is connected and hence $T$ is connected. Further, the graph $T$ is the minimum spanning tree of the graph $G''$, which proves $T$ is acyclic. Hence, $T$ is a tree.
\item[claim 2:] $V(G) = \bigcup_{t \in V(T)} X_t$. \\
Let us partition the vertex set of $G$ into $V_1$ and $V_2$, where $V_1$ denotes the set of vertices which takes part in some cycle of $G$ and $V_2$ denotes the set of vertices which does not take part in any cycle of $G$. It is evident from \emph{Steps 13-17} and from \emph{Steps 9-11}, that every element in $V_1$ and $V_2$ is added to $X_t$, for some $t$, respectively. Thus, $V(G) = \bigcup_{t \in V(T)} X_t$. 

\item[claim 3:] For every edge $\{u,v\} \in E(G)$, there is some $t \in V(T)$ such that $u, v \in X_t$. \\
Every edge, which takes part in some cycle of $G$, is added to $X_t$, for some $t$, by means of $P^{ij}$ in \emph{Steps 4-6} and every non-cycle edge is added to $X_t$, for some $t$, by means of $S$ in \emph{Step 7}. Hence, the claim.

\item[claim 4:] For every vertex $u \in V(G)$, the set $\{ t \in V(T)\vert u \in X_t \}$ induces a subtree of the tree $T$.\\
On the contrary, assume that there exist a vertex $u \in V(G)$ such that the set $\{ t \in V(T)\vert u \in X_t \}$ does not induce a subtree of the tree $T$. i.e., the graph induced by the vertex set $\{ t \in V(T)\vert u \in X_t \}$, say $H$, is not connected. Let $H_1,\ldots, H_l$, $l \geq 2$ be the connected components of $H$. Choose a vertex $X_i$ from $H_1$ and $X_j$ from $H_2$. 
\begin{itemize}
\item[$\bullet$] $\vert X_i \vert = \vert X_j \vert = 2$. The weight of the edge $\{X_i, X_j\} \in E(G'')$ is 5 and hence, this edge will not create a cycle. Thus, $\{X_i, X_j\} \in E(T)$, which is a contradiction as $H_1$ and $H_2$ are disjoint connected components in $H$.
\item[$\bullet$] $\vert X_i \vert = 2$ and $\vert X_j \vert = 3$. The weight of the edge $\{X_i, X_j\} \in E(G'')$ is 4 and hence, this edge will not create a cycle. Thus, $\{X_i, X_j\} \in E(T)$, which is a contradiction.
\item[$\bullet$] $\vert X_i \vert = \vert X_j \vert = 3$ and if the weight of the edge $\{X_i, X_j\} \in E(G'')$ is 1. Then, $\{X_i, X_j\} \in E(T)$ since $T$ is a minimum spanning tree of $G''$ and there can not be a cycle in $G''$ where the weights of all edges are 1.
 
\item[$\bullet$] $\vert X_i \vert = \vert X_j \vert = 3$ and if the weight of the edge $e = \{X_i, X_j\} \in E(G'')$ is 2. The edge $\{X_i, X_j\} \notin E(T)$, implies that, the edge $\{X_i, X_j\}$ is part of a cycle and every other edge in the cycle is of weight one or two. Let $P$ be the second shortest path from $X_i$ to $X_j$ in $G'$ and $V(X_i) \cap V(X_j) = \{u,v\}$. 
\begin{itemize}
\item[-] $(P,e)$ is $C_3$, say $(X_i,X_j,X_s)$. \\ By our assumption, $u \notin V(X_s)$. Since, the weight of $\{X_i,X_s\}$ is either $one$ or $two$, $\vert V(X_i) \cap V(X_s)\vert = 2$. Similarly,  $\vert V(X_j) \cap V(X_s)\vert = 2$. Thus, $\vert  V(X_s)\vert = 4$, which is a contradiction to the construction of $G'$.
\item[-] $(P,e)$ is $C_n, n \geq 4$, say $(X_i,X_{1},\ldots, X_{n-2},X_j)$. Since, $P$ is a shortest path and the weight of the edge $e$ is 2, $V(X_i) \cap V(X_{n-2}) = \emptyset$ and $V(X_j) \cap V(X_1) = \emptyset$. Thus, $u,v \notin V(X_1)$ and $u,v \notin V(X_{n-2})$. The weight of the edge $\{X_i,X_1\}$ is either 1 or 2, implies that, $\vert V(X_i) \cap V(X_1)\vert = 2$, which is a contradiction.
\end{itemize}

\item[$\bullet$] $\vert X_i \vert = \vert X_j \vert = 3$ and if the weight of the edge $\{X_i, X_j\} \in E(G'')$ is 3. Thus, the vertices $X_i$ belongs to some $C_k$, say $S_i$ and the vertices in $X_j$ belongs to some $C_k$, say $S_j$, $i \neq j$, and both $S_i$ and $S_j$ has a vertex intersection. The edge $\{X_i, X_j\} \notin E(T)$, implies that, the edge $\{X_i, X_j\}$ is part of a cycle and every other edge in the cycle is of weight one, two or three. Let $P$ be the second shortest path from $X_i$ to $X_j$ in $G'$ and $V(X_i) \cap V(X_j) = \{u\}$.

\begin{itemize}
\item[-] $(P,e)$ is $C_3$, say $(X_i,X_j,X_s)$. \\ By our assumption, $u \notin V(X_s)$. If $\vert V(X_i) \cap V(X_s)\vert = 2$ and $\vert V(X_j) \cap V(X_s)\vert = 2$ or $\vert V(X_i) \cap V(X_s)\vert = 1$ and $\vert V(X_j) \cap V(X_s)\vert = 2$, then $\vert  V(X_s)\vert = 4$, which is a contradiction to the construction of $G'$. If $\vert V(X_i) \cap V(X_s)\vert = 1$ and $\vert V(X_j) \cap V(X_s)\vert = 3$ or $\vert V(X_i) \cap V(X_s)\vert = 3$ and $\vert V(X_j) \cap V(X_s)\vert = 3$, then the cycle belongs to $X_s$, say $S_s$, $i \neq j \neq s$, contradicts the \emph{Theorem \ref{construction}}. The case where $\vert V(X_i) \cap V(X_s)\vert = 1$ and $\vert V(X_j) \cap V(X_s)\vert = 1$ is not possible by the construction of $G'$.

\item[-] $(P,e)$ is $C_n, n \geq 4$, say $(X_i,X_{1},\ldots, X_{n-2},X_j)$. Since, $P$ is a shortest path and the weight of the edge $e$ is 3, $u$ does not belongs to any internal vertices of $P$. If the weight of the edges $\{X_i,X_1\}$ and $\{X_1,X_2\}$ are 1 and 2 or 2 and 1 or 2 and 2, respectively, then there exists an edge $\{X_i,X_2\}$, which is a contradiction to the minimality of $P$. If the weight of the edges $\{X_i,X_1\}$ and $\{X_1,X_2\}$ are 1 and 1 or 1 and 3 or 3 and 1 or 3 and 3, then the cycle $S_j$ contradicts the \emph{Theorem \ref{construction}}. 
\end{itemize}
\end{itemize}

All the above cases gives the contradiction, hence the claim. $\hfill \qed$

\end{description}
\end{proof}

\noindent Now, we present an algorithm which gives a tree decomposition for $SC_k, k =2m+4, m \geq 1$ graphs, where $\max_{t \in V(T)} \vert X_t \vert = 3$.\\

\noindent \textbf{Outline of the algorithm:}
The algorithm first decomposes the graph $G$ into connected components where each component is a cut edge or a $C_k$ or a CAGE. Next, the algorithm finds the tree decomposition for each connected component. Now, the algorithm combine the components based on its intersection and results in a graph $G'$. Finally, the algorithm finds a minimum spanning tree $T$ of $G'$.

\begin{algorithm}
\caption{Tree Decomposition for $SC_{k}, k= 2m+4, m \geq 1$, graphs}
\begin{algorithmic}[1]
\STATE{\textbf{Input:} $SC_k$ graph $G$ with cycles, $k= 2m+4, m \geq 1$.}
\STATE{\textbf{Output:} Tree decomposition of $G$}
\STATE{Let $\{v_1, \ldots, v_n\}$ be the vertex set of $G$ and let $p=1$.}
\STATE{Decompose the graph $G$ into connected components as per \emph{Theorem \ref{decompositionsc2k}} and let $G_1, \ldots, G_s$ be the connected components in the decomposition.}
\FOR{$i=1$ to $s$}
	\IF{$G_i$ is an edge}
		\STATE{$X_p = V(G_i)$ and $p=p+1$}
	\ELSIF{$G_i$ is a $C_k$}
		\STATE{Let $G_i = (x_{1}, \ldots, x_{k})$ be an induced $C_k$ in $G$ define $X_p$ as follows.
\begin{itemize}
\item[$\bullet$] $X_p = \{x_1,x_2,x_3\}$ and $p=p+1$
\item[$\bullet$] $X_p = \{x_1,x_3,x_4\}$ and $p=p+1$
\item[$\bullet$] $X_p = \{x_1,x_4,x_5\}$ and $p=p+1$, $\ldots$, $X_p = \{x_1,x_{k-1},x_k\}$ and $p=p+1$.
\end{itemize}}
	\ELSIF{$G_i$ is a CAGE}
		\STATE{Collect the vertices in $G_i$ whose degree is equal to $\Delta (G_i)$. CAGE has exactly two such vertices, say $w,z$.}
		\STATE{Let $\Delta (G_i) = s$. Then, $w$ and $z$ have $s$ distinct paths of length $\frac{k}{2}+1$ in $G_i$. Let $(u_{1}^{j}, \ldots, u_{\frac{k}{2}-1}^{j})$ be the $j^{th}$ path between $w$ and $z$, $1 \leq j \leq s$.}
			\FOR{$j=1$ to $s$}
				\STATE{Define $X_p$ as follows:
\begin{itemize}
\item[$\bullet$] $X_p = \{w,z,u_{1}^{j}\}$ and $p=p+1$
\item[$\bullet$] $X_p = \{z,u_{1}^{j}, u_{2}^{j}\}$ and $p=p+1$
\item[$\bullet$] $X_p = \{z,u_{2}^{j}, u_{3}^{j}\}$ and $p=p+1$, $\ldots$, $X_p = \{z,u_{\frac{k}{2}-1}^{j}, u_{\frac{k}{2}}^{j}\}$ and $p=p+1$
\end{itemize}		
		}
			\ENDFOR
\ENDIF
\ENDFOR
\STATE{For each $X_i$, $s+1\leq i \leq s+(k-2)l$, relabel the vertices in $X_i$ by the labels given for the vertices during the input.}

\STATE{Construct a graph $G'$ with vertex set, $V(G') = \{X_i \mid 1 \leq i \leq p-1\}$ and two vertices $X_i, X_j \in V(G')$, $i \neq j$ are adjacent if any one of the following types is true:
\begin{itemize}
\item[Type 1:] $\vert X_i \vert = \vert X_j \vert = 3$ and $X_i, X_j \in V(G_i)$, for some $1 \leq m \leq s$, then, $\vert V(X_i) \cap V(X_j) \vert = 2$. Let $F_1,\ldots, F_s$ be the connected components of the graph after augmenting Type 1 edges. 
\item[Type 2:] $\vert X_i \vert = \vert X_j \vert = 3$ and $X_i \in V(F_r), X_j \in V(F_m)$, for some $1 \leq m, r \leq s$ and $m \neq r$ and $\vert E(G_m) \cap E(G_r) \vert = 1$, and if there are no edges between the vertices of $F_m$ and $F_r$ then, $\vert V(X_i) \cap V(X_j) \vert = 2$.
\item[Type 3:] $\vert X_i \vert = \vert X_j \vert = 3$ and $X_i \in V(F_r), X_j \in V(F_m)$, for some $1 \leq m, r \leq s$ and $m \neq r$ and $\vert V(G_m) \cap V(G_r) \vert = 1$, and if there are no edges between the vertices of $F_m$ and $F_r$ then, $\vert V(X_i) \cap V(X_j) \vert = 1$.
\item[Type 4:] $\vert X_i \vert = 3$ and $\vert X_j \vert = 2$ and $\vert V(X_i) \cap V(X_j) \vert = 1$.
\item[Type 5:] $\vert X_i \vert = 2$ and $\vert X_j \vert = 2$ and $\vert V(X_i) \cap V(X_j) \vert = 1$.
\end{itemize}}
\STATE{Convert the unweighted graph $G'$ to a weighted graph $G''$ by assigning the weight $i$ for the edges of type $i$, $1 \leq i \leq 5$.}
\STATE{Find a minimum spanning tree $T$ for the weighted graph $G''$}
\STATE{\emph{Return $T$}}
\end{algorithmic}
\label{treedecompositionsc2k}
\end{algorithm}

\newpage

\noindent \textbf{\large{Trace of the algorithm}}\\
\noindent We trace the steps of the Algorithm \ref{treedecompositionsc2k} in \emph{Figure \ref{fig:tracetreedecompositionsc2k}}.

\begin{figure}[H]
\centering
\includegraphics[scale=0.26]{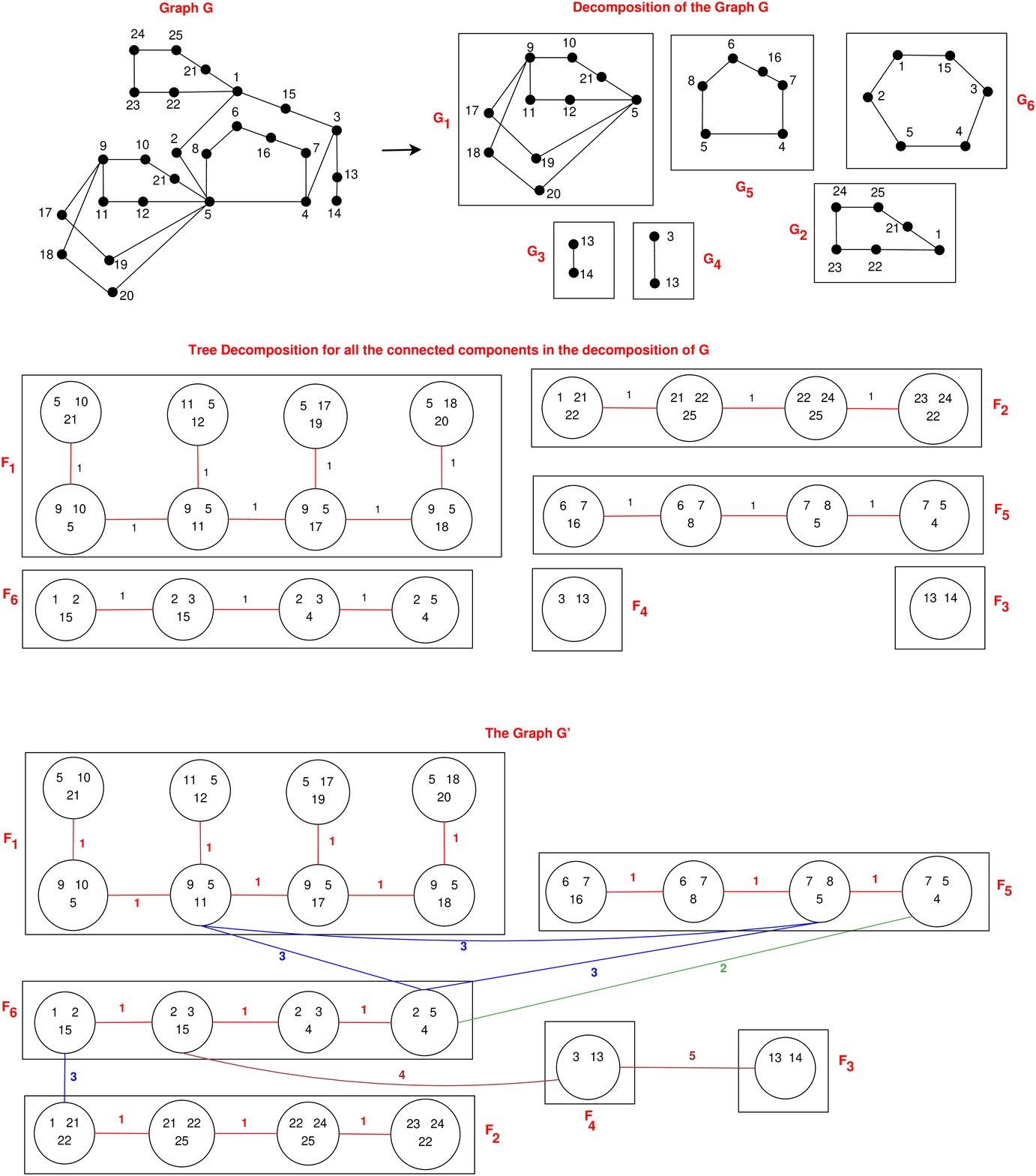} 
\vspace{-0.2cm}
\end{figure}

\begin{figure}[h]
\centering
\includegraphics[scale=0.26]{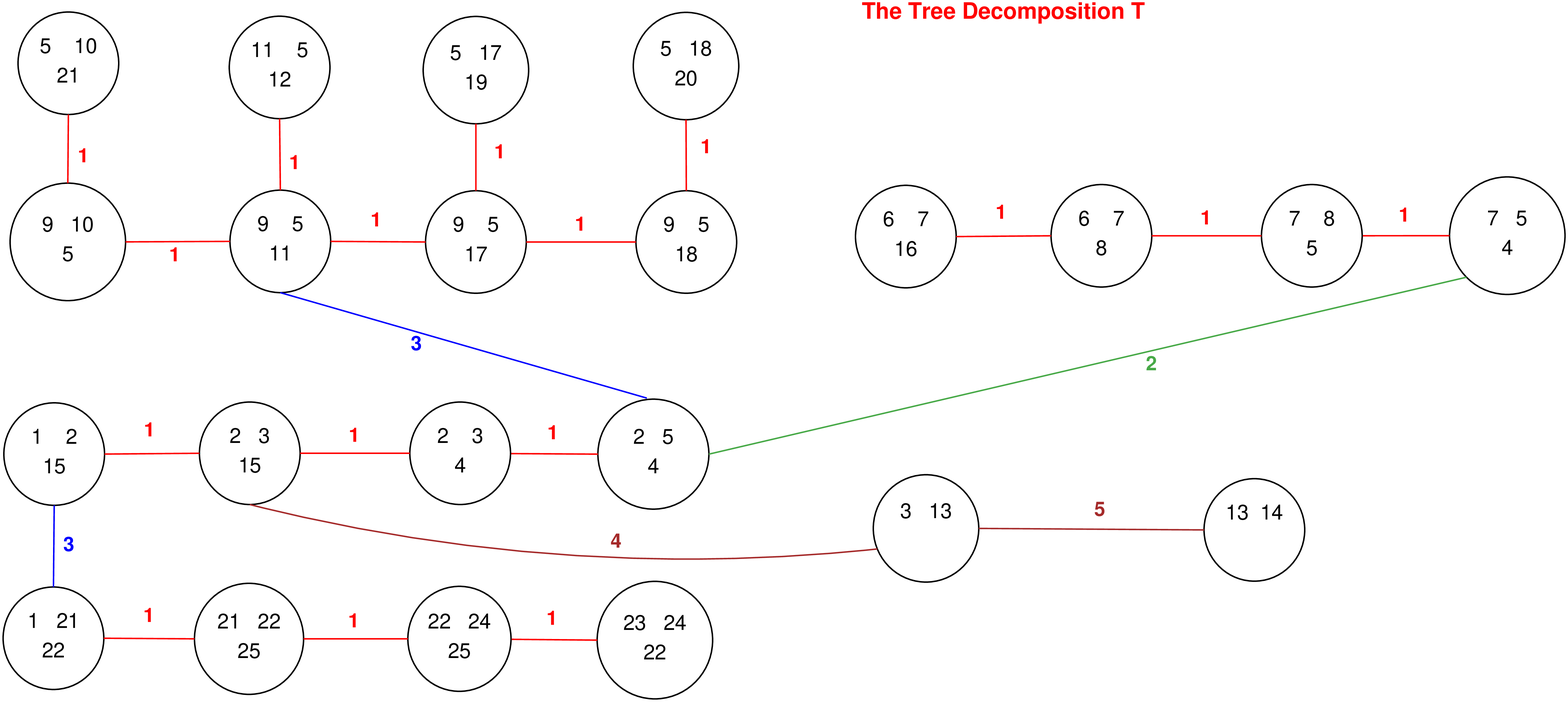} 
\vspace{-0.1cm}
\caption{Tree decomposition of an $SC_6$ graph.}
\label{fig:tracetreedecompositionsc2k}
\end{figure}

\newpage

\begin{theorem}
The graph, $T$, obtained from the \emph{Algorithm \ref{treedecompositionsc2k}} is a tree decomposition of $G$ such that $tw(G) = 2$.
\label{prooftreedecomp2k}
\end{theorem}
\begin{proof}
In the algorithm, we decompose the graph $G$ into connected components, where each connected component is a cut edge or a $C_k$ or a CAGE. It is clear that, for each connected component, the graph constructed in \emph{Steps 5-20} is a tree decomposition of the respective component. Now, we add edges between components based on the conditions in \emph{Step 20} and we make the unweighted graph into a weighted graph $G''$ by giving weights to the edges. Finally, minimum spanning tree $T$ is computed for the graph $G''$. The proof for $T$ is a tree decomposition is similar to the proof in \emph{Theorem \ref{prooftreedecomp2k+1}}. Note that, $\max_{t \in V(T)} \vert X_t \vert = 3$ by \emph{Steps 5-15}.$\hfill \qed$
\end{proof}

\begin{corollary}
\label{partial}
Let $G$ be a connected $SC_k$, $k\geq 5$, graph. Then, $G$ is a partial-2-tree.
\end{corollary}
\begin{proof}
Trivially follows from \emph{Theorem \ref{prooftreedecomp2k+1}} and \emph{Theorem \ref{prooftreedecomp2k}}.$\hfill \qed$
\end{proof}

\begin{theorem}
\label{coloring}
Let $G$ be a connected $SC_k, k \geq 5$, graph. The chromatic number of $G$ is at most \emph{three}. i.e., $\chi(G) \leq 3$. Further, if $k$ is odd then $\chi(G) = 3$ and  if $k$ is even then $\chi(G) = 2$.
\end{theorem}
\begin{proof}
If $k$ is even, then $G$ is bipartite and hence, $\chi(G) \leq 2$. If $k$ is odd: 
let $S$ be the maximum independent set in the graph induced on the non-tree edges of T. From \emph{Lemma \ref{matching}} (\emph{Section 5.2}), it follows that the set of non-tree edges in T forms a matching. Thus, $\chi (G\backslash S) \leq 2$ and $S$ can be colored using the third color. Hence, $G$ requires at most three colors. Therefore, we can conclude $\chi(G) \leq 3$ if $k$ is odd. We can also prove the theorem from the fact that $SC_k$, $k\leq 5$, graphs are partial-2-trees.  $\hfill \qed$
\end{proof}

\section{Conclusions and Further Research}
In this paper, we have investigated strictly chordality $k$ graphs, graphs in which every induced cycle is of length $k$ or cycle-free, from both structural and algorithmic perspectives. We have obtained nice structural results based on the structure of the minimal vertex separators. Further, we have shown that testing $SC_k$ graphs are polynomial-time solvable using a special ordering, namely Vertex Cycle Ordering (VCO). Other results include Coloring, Hamiltonicity and Treewidth. Classical problems such as Vertex Cover, Odd Cycle Transversal, Feedback Vertex Set etc., are yet to be explored restricted to $SC_k$ graphs. 


\nocite{*}
\begin{thebibliography}{}

\bibitem{app1}
Jean R. S. Blair and Barry Peyton: An Introduction to Chordal Graphs and Clique Trees. In Graph Theory and Sparse Matrix Computation, Vol.56, pp.1-29, (1993).

\bibitem{app2}
I. Duff and J. Reid: The multifrontal solution of indefinite sparse symmetric linear equations. ACM Trans math. Software, Vol.9 , pp. 302 - 325, (1983).

\bibitem{app3}
R.P. Anstee, M. Farber, Characterizations of totally balanced matrices: Journal of Algorithms, Vol.5, pp.215-230, (1984).

\bibitem{Hajnal}
A. Hajnal and T. Surányi: $\ddot{U}$ber die Aufl$\ddot{o}$sung von Graphen vollstandiger Teilgraphen, Annales Universitatis Scientarium Budapestinensis de Rolando Eötvös Nominatae
Sectio Mathematica. Math.,1 (1958).


\bibitem{dirac}
G.A. Dirac: On rigid circuit graphs. Abhandlungen aus dem Mathematischen Seminar der Universität Hamburg, Vol.25, pp.71-76, (1961).

\bibitem{fulkerson}
D.R. Fulkerson and O.A. Gross: Incidence matrices and interval graphs. Pacific Journal of Mathematics, Vol.15, pp.835-855, (1965).

\bibitem{tarjan}
Donald J.Rose,  George Lueker and R.E. Tarjan: Algorithmic aspects of vertex elimination on graphs. SIAM Journal of Applied Mathematics, Vol.34, pp.176-197, (1978).

\bibitem{GolumbicGoss}
M.C. Golumbic, C.F. Goss: Perfect elimination and chordal bipartite graphs. Journal of Graph Theory, Vol.2, pp.155-163, (1978).

 \bibitem{gavril}
F$\check{a}$nic$\check{a}$ Gavril: The intersection graphs of subtrees in trees are exactly the chordal graphs. Journal of Combinatorial Theory, Series B, Vol.16, pp.47-56, (1974).

\bibitem{corneil}
D.G. Corneil and J. Fonlupt: The complexity of generalized clique covering. Discrete Applied Mathematics, Vol.22, pp.109-118, (1989). 

\bibitem{Hoang}
C.T. Hoang: Efficient algorithms for minimum weighted coloring of some classes of perfect graphs. Discrete Applied Mathematics, Vol.55, pp.133-143, (1994).



\bibitem{bodlaender}
H.L. Bodlaender: A tourist guide through treewidth. Acta Cybernetica, Vol.11, pp.1-23 (1993).

\bibitem{kloks}
T. Kloks and D. Kratsch: Treewidth of Chordal Bipartite Graphs. Technical report, Utrecht University, (1992).




\bibitem{booth}
K.S. Booth and J.H. Johnson: Dominating sets in chordal graphs. SIAM Journal of Computation, Vol.11, pp.191-199, (1982).

\bibitem{mueller}
Haiko Mueller and Andreas Brandstaedt: The NP-completeness of Steiner tree and dominating set for chordal bipartite graphs. Theoretical Computer Science, Vol.53, pp.257-265, (1987).

\bibitem{colbourn}
C.J. Colbourn and L.K. Stewart: Dominating cycles in series-parallel graphs. Ars Combinatoria. 19A, pp.107-112, (1985).


\bibitem{muellerh}
H. Mueller: Hamiltonian circuits in chordal bipartite graphs. Discrete Mathematics, Vol.156, pp.291-298, (1996).

\bibitem{join}
Pavol Hell and Pei-Lan Yen: Join colourings of chordal graphs. Discrete Mathematics, Vol.338, pp.2453-2461, (2015).


\bibitem{contractibility}
Remy Belmonte, Petr A. Golovach, Pinar Heggernes, Pim van't Hof, Marcin Kaminki, and Daniel Paulusma: Detecting Fixed Patterns in Chordal Graphs
in Polynomial Time. Algorithmica, Vol.69, pp. 501-521, (2014).

\bibitem{strong}
Ton Kloks, Sheung-Hung Poon, Chin-Ting Ung and  Yue-Li Wang: On the strong chromatic index and maximum induced matching of tree-cographs, permutation graphs and chordal bipartite graphs. Journal of Discrete Algorithms, Vol. 30, pp. 21-28, (2015). 

\bibitem{enumeration}
Petr A. Golovach, Pinar Heggernes, Mamadou M. Kante, Dieter Kratsch and Yngve Villanger: Enumerating minimal dominating sets in chordal bipartite graphs. Discrete Applied Mathematics (to be published), (2015).

\bibitem{reconfiguration}
Marthe Bonamy, Matthew Johnson, Ioannis Lignos, Viresh Patel and Daniel Paulusma: Reconfiguration graphs for vertex colourings of chordal and chordal bipartite graphs. Journal of Combinatorial Optimization, Vol.27, pp. 132-143, (2012).


\bibitem{golumbicbook}
M.C.Golumbic: Algorithmic Graph Theory and Perfect Graphs. Academic Press, New York, (1980).


\bibitem{dbwest}
D.B.West: Introduction to Graph Theory. Published by Prentice Hall, (2001).

\bibitem{golumbic}
M.C.Golumbic: Dirac's theorem on triangulated graphs. Annals of the New York Academy of Sciences. 319, pp.242-246, (1979).

\bibitem{chvatal}
Chvatal V: In the travelling salesman problem: A guided tour of combinatorial optimization. Wiley, pp. 403-429, (1985).

\end{thebibliography}
\end{document}